\documentclass{vldb}

\pdfoutput=1

\usepackage{times}
\usepackage{graphicx}
\usepackage{float}
\usepackage{amssymb}
\usepackage{amsmath}
\usepackage{amsfonts}
\usepackage{algorithm}
\usepackage{array}
\usepackage{color}

\newtheorem{theorem}{Theorem}
\newtheorem{lemma}{Lemma}

\newtheorem{definition}{Definition}
\DeclareMathOperator{\EXP}{Exp}
\DeclareMathOperator{\POLY}{Poly}
\DeclareMathOperator{\PR}{Pr}
\DeclareMathOperator{\LAP}{Lap}
\DeclareMathOperator{\TUP}{\text{-}Tup}
\DeclareMathOperator{\SUPP}{supp}
\DeclareMathOperator{\TRUE}{True}
\DeclareMathOperator{\FALSE}{False}
\DeclareMathOperator{\IMPACT}{impact}
\DeclareMathOperator{\AVGDEG}{avgdeg}

\title{Recursive Mechanism: Towards Node Differential Privacy and Unrestricted Joins [Full Version, Draft 0.1]\thanks{A different version of this paper has been accepted by SIGMOD 2013.}}

\author{
    Shixi Chen\\
    Fudan University\\
    chensx@fudan.edu.cn
\and
    Shuigeng Zhou\\
    Fudan University\\
    sgzhou@fudan.edu.cn
}

\allowdisplaybreaks

\begin{document}
    \maketitle

    \begin{abstract}
Existing studies on differential privacy mainly consider aggregation on data sets where each entry corresponds to a particular participant to be protected. In many situations, a user may pose a relational algebra query on a sensitive database, and desires differentially private aggregation on the result of the query. However, no known work is capable to release this kind of aggregation when the query contains unrestricted join operations. This severely limits the applications of existing differential privacy techniques because many data analysis tasks require unrestricted joins. One example is subgraph counting on a graph.
Existing methods for differentially private subgraph counting address only edge differential privacy and are subject to very simple subgraphs. Before this work, whether any nontrivial graph statistics can be released with reasonable accuracy under node differential privacy is still an open problem.

In this paper, we propose a novel differentially private mechanism to release an approximation to a linear statistic of the result of some positive relational algebra calculation over a sensitive database. Unrestricted joins are supported in our mechanism. The error bound of the approximate answer is roughly proportional to the \emph{empirical sensitivity} of the query --- a new notion that measures the maximum possible change to the query answer when a participant withdraws its data from the sensitive database. For subgraph counting, our mechanism provides the first solution to achieve node differential privacy, for any kind of subgraphs.
\end{abstract}

    \section{Introduction}\label{sec:intro}

An important task in data privacy research is to develop mechanisms to publish useful results mined from sensitive database, without disclosing individual privacy. Most of existing techniques provide rather limited privacy protection, since they usually address specific attack models, or rely on specific assumptions about the prior knowledge the potential adversary may possess. In recent years, the paradigm of differential privacy has received increasing attention, because it can provide robust and quantitative privacy guarantee while making no assumptions about the prior knowledge of the adversary. Data publishing algorithms that achieve differential privacy should guarantee that their outputs are randomized such that input databases differing in one participant are almost indistinguishable to the adversary. Therefore, participating in a database is unlikely to cause privacy breach.

Existing studies on differential privacy are mainly based on a simple data model, where the input database is a set of records, and each record corresponds to a participant. The output of a differentially private data publishing algorithm should have almost identical probability distributions for input data sets that differ in exactly one record. Various kinds of queries that compute aggregations on data sets have been considered, and much effort has been put to linear aggregations, on which more complex queries can be built.

The success of most existing differentially private mechanisms relies on the precondition that the maximum possible change to the query answer resulted from the change of one participant should be small and bounded. Such maximum possible change is called the \emph{sensitivity} of the query, which determines the minimum magnitude of noise needed to introduce into the answer. In practice, however, many databases contain information about not only individual participants, but also relationships between them. The change of one participant may, in the worst case, have potentially unlimited impact on the database and the query answer. Queries on such databases are too complex to be tackled by existing techniques. In this paper, we try to relax the precondition by allowing potentially unbounded impact that may be incurred by new participants joining the database, and give an elegant solution.

\subsection{Motivation}

Subgraph counting is an important problem in data mining and social networks, which counts the number of occurrences of a given query subgraph in an input graph. Despite of the enormous works on anonymization schemes for private graphs, little has been down to provide quantitative guarantees of privacy and utility. In~\cite{DBLP:conf/pods/RastogiHMS09}, subgraph counting is studied under a much weaker version of differential privacy. Their privacy guarantee protects only against a specific class of adversaries. The error of the approximate answer returned by their algorithm is large --- the magnitude of noise grows exponentially with the number of edges in the subgraph. In~\cite{DBLP:conf/stoc/NissimRS07} and \cite{DBLP:journals/pvldb/KarwaRSY11}, $k$-triangle and $k$-star counting are studied, and they achieve better privacy and utility guarantee. In particular, they achieve $\epsilon$-differential privacy for $k$-star counting, and $(\epsilon,\delta)$-differential privacy, a weaker version of differential privacy, for $k$-triangle counting. However, their work cannot be extended to other kinds of subgraph. It is also worthy of mentioning the work in \cite{DBLP:conf/icdm/HayLMJ09}, which gives an algorithm for releasing an approximation to the degree distribution
of a graph and achieves $k$-edge differential privacy.

A major problem of the above works is that they can only achieve edge privacy --- each edge corresponds to a participant to be protected. But for many real-world data, such as social networks, each individual participant contributes to the graph a node rather than just an edge. We desire privacy protection based on nodes rather than edges. Unfortunately, it is difficult to achieve node differential privacy while obtaining reasonable query accuracy, because the maximum possible change to the query answer resulted from the change of one node~(as well as all of its incident edges) is comparable to the graph size. Prior to our work, whether any nontrivial graph statistics can be released under node differential privacy with reasonable accuracy is still an open problem~\cite{DBLP:journals/pvldb/KarwaRSY11}. It was widely believed that algorithms achieving node differential privacy can only return query answers that are too noisy for practical applications~\cite{DBLP:conf/sigmod/KiferM11,DBLP:conf/icdm/HayLMJ09}. In this paper, we try to challenge this seemingly impossible task and give a general solution.

In reality, databases usually consist of a number of tables. A participant may contribute tuples to several tables, and a tuple can be contributed collectively by multiple participants. A user may want to issue a SQL query to the database to obtain an output table, then requests approximate statistic of the output table. Subgraph counting is, in fact, a special case of this general context, because every subgraph count can be written as a SELECT query. It will be quite useful if this kind of task can be solved under differential privacy. There have been at least two attempts in the literature~\cite{DBLP:conf/sigmod/McSherry09,Palamidessi12}, which are based on bounding the \emph{global sensitivity} of the query. However, these works support only restricted kinds of join operations, where one participant can affect only constant number of tuples in the output table. Even the most simple subgraph counting requires unrestricted joins where a participant can have unbounded impact on the query answer. Obviously, existing methods are unable to support this kind of joins.

We focuses on the case where the SQL query can be translated into a series of positive relational algebra calculation. We aim at releasing an approximation to a linear statistic of the output table with reasonable accuracy under differential privacy. Our solution covers subgraph counting. Both node and edge differential privacy are achievable, depending on the choice of user. Node differential privacy is stronger than edge differential privacy, but the latter can allow better query accuracy. When nodes or edges of the graph are associated with auxiliary information, our solution also allows arbitrary kinds of constraints imposed on any edges or nodes of the subgraph, which are not supported by prior works.

\subsection{Contributions}

To develop differentially private mechanisms that can support unrestricted joins, we face several difficulties. First, the problem we study allows one participant to have complex impact on the database. The data model assumed by existing differentially private mechanisms is too simple to suffice our need to express the complex relations between the database and the participants. Hence, new data model is needed to express how participants affect the database content. Second, existing notions of \emph{sensitivity}, including global and local sensitivity, are no longer appropriate in our case, because a new participant joining the database can, in the worst case, have unlimited impact on the query answer, leading to unbounded sensitivity. Thus, it is impossible for us to calibrate the noise to such sensitivities. We need a new metric to measure the least magnitude of noise that is necessary to answer a query. Third, existing works for complex queries often compromise privacy guarantee, utility guarantee or efficiency guarantee. However, such compromise can lead to severe problem for practical use, which limit the applications of those techniques. It is a challenging task to develop mechanisms that can achieve all three guarantees.

Contributions of this paper are as follows:

1) We propose a general model of \emph{sensitive databases}, which allows one participant to affect the database content in any possible way. By formalizing the definition of \emph{neighborhood}, the notion of differential privacy on this data model is setting up such that privacy protection is based on individual participants.

2) We propose a new notion of sensitivity, called \emph{empirical sensitivity}, that measures the maximum possible change to the query answer when a participant withdraws its data from the current database content. Empirical sensitivity is always bounded, and is often small. It gives a better measure of the least magnitude of noise that is necessary to answer a query.

3) We develop a general but inefficient mechanism to answer any monotonic query on a sensitive database. This mechanism guarantees $\epsilon$-differential privacy, and the error bound is roughly proportional to the \emph{global empirical sensitivity} of the query.

4) We propose a specific model of \emph{sensitive databases} based on $K$-relation or $c$-table. Every tuple in a $K$-relation is annotated with a positive Boolean expression that specifies its condition of presence. $K$-relation is closed under positive relational algebra calculation. Hence it can be used to express the complex relations between the participants and the table output by a SQL query.

5) We develop an efficient mechanism to answer any linear query to a sensitive $K$-relation. This mechanism guarantees $\epsilon$-differential privacy, and the error bound is roughly proportional to the \emph{universal empirical sensitivity} of the query. The computation cost is in a polynomial of the size of $K$-relation. Our mechanism is the first solution to the problem of subgraph counting for any subgraphs, which can achieve either node differential privacy or edge differential privacy, and the error bound is roughly proportional to the \emph{local empirical sensitivity} of the query.

6) We conduct extensive experiments to evaluate the proposed mechanism. Experimental results validate the effectiveness and efficiency of the new mechanism.

In Fig.~\ref{fig:intro:compare} we present a brief comparison between our mechanism and existing mechanisms.

\begin{figure*}[!t]
    \centering
    \begin{tabular}{|m{0.32\textwidth}|m{0.3\textwidth}|m{0.32\textwidth}|}
    \hline
    Queries & Our mechanism & Existing mechanisms \\\hline
    Monotonic query on a sensitive database &
    $\widetilde{O}(\widetilde{GS}_q/\epsilon)$ error, $\EXP(|P|)$ time &
    None \\\hline

    Linear statistic of the output of a SQL query &
    $\widetilde{O}(\widetilde{US}_q/\epsilon)$ error, $\POLY(|P|,|R|)$ time &
    \parbox{.32\textwidth}{$O(US_q/\epsilon)$ error and $O(1)$ time if there are no unrestricted joins~\cite{DBLP:conf/sigmod/McSherry09,Palamidessi12}\\ Not solvable if there are unrestricted joins because $US_q\geq GS_q=+\infty$} \\\hline

    triangle counting (\includegraphics[height=10pt]{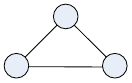}) &
    $\widetilde{O}(\widetilde{LS}_q/\epsilon)$ error, $\POLY(k,|R|)$ time &
    $O(LS_q/\epsilon+1/\epsilon^2)$ error; $O(|V|\cdot|E|)$ time; only achieve differential privacy based on edges~\cite{DBLP:conf/stoc/NissimRS07} \\\hline

    $k$-star counting (e.g., 3-star \includegraphics[height=10pt]{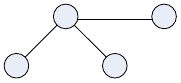}) &
    $\widetilde{O}(\widetilde{LS}_q/\epsilon)$ error, $\POLY(k,|R|)$ time or $\POLY(|V|,|E|,k)$ time &
    $O(LS_q/\epsilon)$ error if $1/\epsilon=O(d_{\max}/k)$; $O(|V|\cdot|E|)$ time; only achieve differential privacy based on edges~\cite{DBLP:journals/pvldb/KarwaRSY11} \\\hline

    $k$-triangle counting (e.g., 3-triangle \includegraphics[height=10pt]{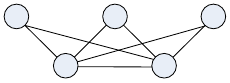}) &
    $\widetilde{O}(\widetilde{LS}_q/\epsilon)$ error, $\POLY(k,|R|)$ time or $\POLY(|V|,|E|,k)$ time &
    $O(LS_q/\epsilon)$ error if $\ln(1/\delta)/\epsilon=O(a_{\max})$; $O(|V|\cdot|E|)$ time; only achieve $(\epsilon,\delta)$-differential privacy based on edges~\cite{DBLP:journals/pvldb/KarwaRSY11} \\\hline

    $k$-node $l$-edge connected subgraph counting &
    $\widetilde{O}(\widetilde{LS}_q/\epsilon)$ error, $\POLY(k,l,|R|)$ time &
    $\Theta((k l^2\log |V|)^{l-1}/\epsilon)$ error; $O(1)$ time; only achieve adversary privacy based on edges w.r.t. a specific class of adversaries~\cite{DBLP:conf/pods/RastogiHMS09} \\\hline

    \end{tabular}

    \caption{Comparison between our mechanism and existing mechanisms. $\widetilde{O}$ means that logarithmic factors are omitted. For sensitive database, $|P|$ denotes the number of participants and $|\SUPP(R)|$ denotes the number of tuples returned by the SQL query. For subgraph counting, $|V|$ and $|E|$ denote the number of nodes and edges in the graph, and $|R|=|\SUPP(R)|$ denotes the true query answer. $d_{\max}$ denotes the maximum degree of a node, and $a_{\max}$ denotes the maximum number of common neighbors of a pair of nodes. $GS$, $LS$, $US$, $\widetilde{GS}$, $\widetilde{LS}$ and $\widetilde{US}$ are explained in Sec.~\ref{sec:pre} and Sec.~\ref{sec:problem}. We have $\widetilde{LS}_q\leq LS_q$ and $\widetilde{US}_q\leq US_q$. Note that we do not take account of the time needed for generating the output table or the list of matched subgraphs in the computation cost. For subgraph counting our solution can achieve differential privacy based on either nodes or edges, depending on the choice of user.}
    \label{fig:intro:compare}
\end{figure*}

    \section{Preliminaries}
\label{sec:pre}

\subsection{Privacy and Utility}

In this work, we will use differential privacy~\cite{DBLP:conf/tcc/DworkMNS06}, a state-of-the-art paradigm for privacy preserving data publishing. A randomized algorithm is differentially private if it yields nearly identical distributions over its outcomes when running on neighboring databases.

\begin{definition}[Differential Privacy]
    A randomized algorithm $\mathcal{A}$ is $(\epsilon,\delta)$-differentially private if for any pair of neighboring databases $D$, $D'$, and for any set of possible outputs $S\subseteq Range(\mathcal{A})$,
    \begin{equation}
        \PR[\mathcal{A}(D)\in S]\leq e^\epsilon \cdot \PR[\mathcal{A}(D')\in S]+\delta
    \end{equation}
    where the probability is taken over the randomness of $\mathcal{A}$. When $\delta=0$, the algorithm is $\epsilon$-differentially private.
\end{definition}
All algorithms presented in this paper satisfy $\epsilon$-differential privacy.

The definition of \emph{neighboring} depends on the context or application. Usually, $D$ and $D'$ are said to be neighboring if they differ only by one participant. In this case, a differentially private algorithm can protect against disclosure of any participant. In the literature, the database $D$ is often considered as a multiset of records, where each record corresponds to a particular participant, then $D$ and $D'$ are neighboring if $|D-D'|+|D'-D|=1$.

We are interested in queries that are real-valued functions of the database (though other kinds of queries are also important). A differentially private algorithm must introduce randomness to its output, and the answer is never exact. The utility of the algorithm is measured by how accurate its answer is. 

\begin{definition}[$(\epsilon,\delta)$-Accurate]
    For a database $D$, a query $q$ and the true answer $q(D)$, we say that the answer returned by an algorithm $\mathcal{A}$ is $(\epsilon,\delta)$-accurate if
    \begin{equation}
        \PR[|\mathcal{A}(D)-q(D)|>\epsilon] \leq \delta
    \end{equation}
\end{definition}

\subsection{Global Sensitivity}

A well known approach to achieve differential privacy is Laplace mechanism~\cite{DBLP:conf/tcc/DworkMNS06}, which introduces i.i.d. noises into the query answers. The magnitude of noise is calibrated to the \emph{sensitivity} of the query --- a property of the query that measures the maximum possible change to the true answer caused by a small change in the database.

\begin{definition}[Global Sensitivity]
    For a real-valued function $q:\mathbb{D}\rightarrow\mathbb{R}^m$, the (global) sensitivity of $q$ is
    \begin{equation}
        GS_q=\max_{D,D'\in\mathbb{D}}\|q(D)-q(D')\|_1
    \end{equation}
    where the maximum is taken over all pairs of neighboring databases $D$, $D'$.
\end{definition}

Given a database $D\in\mathbb{D}$, a query sequence $q:\mathbb{D}\rightarrow\mathbb{R}^m$ and a parameter $\epsilon>0$, Laplace mechanism $\mathcal{A}$ returns $\mathcal{A}(D)=q(D)+(Y_1,\dots,Y_m)$, where $Y_i$ are i.i.d. random variables that follow Laplace distribution $\LAP(GS_q/\epsilon)$, which has the following probability density function
\begin{equation}
    \LAP(y|b)=\frac{1}{2b}\exp(-\frac{|y|}{b})
\end{equation}

Laplace mechanism satisfies $\epsilon$-differential privacy. It is easy to verify that Laplace mechanism returns $(c GS_q/\epsilon,e^{-c})$-accurate answer to each query in the sequence $q$, for any $c>0$.

\subsection{Local Sensitivity and Smooth Sensitivity}

In Laplace mechanism, the magnitude of noise depends on $GS_q$ and the parameter $\epsilon$, but not on the database $D$. Since the global sensitivity $GS_q$ measures the impact of a participant on the true answer in the worst case, this often introduces unnecessarily large noise. In \cite{DBLP:conf/stoc/NissimRS07}, a local measure of sensitivity was proposed

\begin{definition}[Local Sensitivity]
    For a real-valued function $q:\mathbb{D}\rightarrow\mathbb{R}^m$ and a database $D\in\mathbb{D}$, the local sensitivity of $q$ at $D$ is
    \begin{equation}
        LS_q(D)=\max_{D'}\|q(D)-q(D')\|_1
    \end{equation}
    where the maximum is taken over the neighborhood of $D$.
\end{definition}

Observing that $GS_q=\max_D LS_q(D)$, we know that $LS_q(D)$ never exceeds $GS_q$. Ideally, we would like to release $q(D)$ with noise magnitude proportional to $LS_q(D)$, but the noise magnitude might leak information and differential privacy is not satisfied. \cite{DBLP:conf/stoc/NissimRS07} proposed that the noise magnitude should be calibrated to a smooth upper bound $S$ on the local sensitivity, namely, a function $S$ that is an upper bound on $LS_f$ at all point and such that $\ln(S(\cdot))$ has low global sensitivity. \cite{DBLP:conf/stoc/NissimRS07} presents algorithms to compute the optimal $S$, called the \emph{smooth sensitivity} of $q$, for a variety of queries.

\subsection{$K$-Relation and $c$-Table}
\label{sec:pre:krelation}

Our work addresses aggregation on relations where each tuple could be contributed by multiple participants, and each participant could contribute multiple tuples. To track which participants contribute a tuple and how they contribute, we use $K$-relation~\cite{DBLP:conf/pods/GreenKT07} or $c$-table~\cite{DBLP:journals/jacm/ImielinskiL84}, a model proposed in the field of uncertain databases, where tuples are annotated (tagged) with their provenance information, and positive relational algebra is generalized to such tagged-tuple relations. Here we briefly review $K$-relation and $c$-table.

Let $U$ be a finite set of attributes and $\mathbb{C}$ a domain of values, then each tuple is a function $t:U\rightarrow\mathbb{C}$. The set of all such $U$-tuples is denoted by $U\TUP$. Relations without annotations are just subsets of $U\TUP$. Tuples in a $K$-relation are annotated with elements from a semiring $(K,+,\cdot,0,1)$. A $K$-relation over $U$ is a function $R:U\TUP\rightarrow K$ with a finite support $\SUPP(R)=\{t|R(t)\neq 0\}$. The operations of positive algebra on $K$-relation are defined as follows~\cite{DBLP:conf/pods/GreenKT07}:
\begin{description}
\item[empty relation]
    For any set of attributes $U$, there is $\emptyset:U\TUP\rightarrow K$ such that $\emptyset(t)=0$ for all $t$.
\item[union]
    For $R_1,R_2:U\TUP\rightarrow K$, $R_1\cup R_2:U\TUP\rightarrow K$ is defined by
    \[
        (R_1\cup R_2)(t)=R_1(t)+R_2(t)
    \]
\item[projection]
    For $R:U\TUP\rightarrow K$ and $V\subseteq U$, $\pi_V R:V\TUP\rightarrow K$ is defined by
    \[
        (\pi_V R)(t)=\sum_{\text{$t=t'$ on $V$ and $R(t')\neq 0$}}R(t')
    \]
\item[selection]
    For $R:U\TUP\rightarrow K$ and a selection predicate $P:U\TUP\rightarrow\{0,1\}$, $\sigma_P R:U\TUP\rightarrow K$ is defined by
    \[
        (\sigma_P R)(t)=R(t)\cdot P(t)
    \]
\item[natural join]
    For $R_i:U_i\TUP\rightarrow K$, $i=1,2$, $R_1\bowtie R_2:(U_1\cup U_2)\TUP\rightarrow K$ is defined by
    \[
        (R_1\bowtie R_2)(t)=R_1(t_1)\cdot R_2(t_2)
    \]
    where $t_1=t$ on $U_1$ and $t_2=t$ on $U_2$.
\item[renaming]
    For $R:U\TUP\rightarrow K$ and a bijection $\beta:U\rightarrow U'$, $\rho_\beta R:U'\TUP\rightarrow K$ is defined by
    \[
        (\rho_\beta R)(t)=R(t\circ\beta)
    \]
\end{description}
\emph{Intersection} and \emph{cartesian product} are just special cases of natural join. But \emph{difference} is not supported in positive relational algebra.

We study differentially private aggregation on a $c$-table, which is a special case of $K$-relation where $K$ makes up of positive Boolean expressions over some set $B$ of variables. The term \emph{positive} means that the expressions do not involve negation ($\neg$), but only disjunction ($\vee$), conjunction~($\wedge$) and constants $\TRUE$ and $\FALSE$. In our work, each variable in $B$ may correspond to a (potential) participant being protected, then the Boolean expression annotated with a tuple $t$ gives the condition of $t$ being presented in the relation when some participants may opt out.

In $c$-table or $K$-relation, expressions that yield the same truth-value for all valuation of variables in $B$ are considered equivalent. But this is not applicable to our work. An expression $(b_1\vee b_2)\wedge(b_1\vee b_3)$ cannot be simply rewritten into $b_1\vee(b_2\wedge b_3)$. Such rewriting could make our mechanism fail to satisfy differential privacy. We will review this issue later.

    \section{Problem Formulation}
\label{sec:problem}

\subsection{Sensitive Databases and Monotonic Queries}

\label{sec:problem:sendb}

In the literature of differential privacy, a \emph{sensitive database} is typically considered as a multiset of records, and the privacy is defined by the indistinguishability between data sets that differ by only one record. But this definition of privacy is no longer appropriate in our case, where each participant could have complex effect on the database. To achieve differential privacy in our setting, we need to know about not only the content of the database, but also how it changes if some participants withdraw their data. A sensitive database being released should contain such self-descriptive information. We propose a new definition of \emph{sensitive database}, as below, which is more general.

\begin{definition}[Sensitive Database]
    A sensitive database is an ordered pair $(P,M)$, where $P$ is finite set of participants contributing the data, and $M$ is a function $M:\mathcal{P}(P)\rightarrow\mathbb{D}$ such that $M(P')$ is the content of the database if only participants in $P'$ contribute their data.
\end{definition}

Once sensitive databases are formalized, we are ready to adapt the notion of differential privacy to them by making clear what sensitive databases are considered neighboring with each other. We say that two sensitive databases are neighboring if one database can be obtained from the other by one participant withdrawing its data.

\begin{definition}[Neighboring]
    Two sensitive databases $(P_1,M_1)$ and $(P_2,M_2)$ are neighboring if $|P_1-P_2|+|P_2-P_1|=1$ and $M_1(P')=M_2(P')$ for all $P'\subseteq P_1\cap P_2$.
\end{definition}

\begin{definition}[Ancestor]
    We say that $(P_1,M_1)$ is an ancestor of $(P_2,M_2)$, denoted by $(P_1,M_1)\preceq(P_2,M_2)$, if $P_1\subseteq P_2$ and $M_1(P')=M_2(P')$ for all $P'\subseteq P_1$.
\end{definition}

We postulate a class $\Omega$ of sensitive databases, such that every possible sensitive database being considered is an element of $\Omega$. Moreover, if $(P,M)\in\Omega$, then all ancestors of $(P,M)$ are also elements of $\Omega$. We make a further assumption that there is a special element $D_0$ in $\mathbb{D}$ such that $M(\emptyset)=D_0$ for all $(P,M)\in\Omega$ (otherwise, $\Omega$ comprises disconnected parts).

For a sensitive database $(P,M)$, a query $q$ takes as input $M(P)$, the current content of the database, and outputs $q(M(P))$. In this paper, we address queries that output a real number and are monotonic.

\begin{definition}[Monotonic Query]
    For a class $\Omega$ of sensitive databases, a query $q:\mathbb{D}\rightarrow\mathbb{R}$ is monotonic if both of following hold:
    \begin{itemize}
    \item
        $q(D_0)=0$
    \item
        $q(M_1(P_1))\leq q(M_2(P_2))$ for all $(P_1,M_1)\preceq(P_2,M_2)$
    \end{itemize}
\end{definition}

If the global sensitivity of a query is low, then Laplace mechanism can still be applied to obtaining differentially private answer with reasonable accuracy. In many applications, however, the change of a participant could, in the worst case, incur excessive or even unlimited impact on the database content as well as the query answer. No existing differentially private techniques can process queries with unbounded global/local sensitivity. Hence, global/local sensitivity is no longer an appropriate quantity to measure the necessary amount of noise introduced into the query answer. We propose a new notion of sensitivity, \emph{empirical sensitivity}, which suffices our need.

\begin{definition}[Local Empirical Sensitivity]
    For a real-valued function $q:\mathbb{D}\rightarrow\mathbb{R}^m$ and a sensitive database $(P,M)$, the local empirical sensitivity of $q$ at $(P,M)$ is
    \begin{equation}
        \widetilde{LS}_q(P,M)=\max_{p\in P}\|q(M(P))-q(M(P-\{p\}))\|_1
    \end{equation}
    If $P=\emptyset$, then $\widetilde{LS}_q(P,M)=0$.
\end{definition}

\begin{definition}[Global Empirical Sensitivity]
    For a real-valued function $q:\mathbb{D}\rightarrow\mathbb{R}^m$ and a sensitive database $(P,M)$, the global empirical sensitivity of $q$ at $(P,M)$ is
    \begin{equation}
        \widetilde{GS}_q(P,M)=\max_{(P',M')\preceq(P,M)}\widetilde{LS}_q(P',M')
    \end{equation}
\end{definition}

Empirical sensitivity measures the maximum possible change to the query answer when a participant opts out. It is obvious that $\widetilde{LS}_q(P,M)\leq LS_q(M(P))\leq GS_q$ and $\widetilde{LS}_q(P,M)\leq\widetilde{GS}_q(P,M)\leq GS_q$.

\subsection{Linear Queries on Sensitive Relations}

\label{sec:problem:senrel}

Although the model of sensitive databases and monotonic queries is general, it may be too general to allow efficient mechanism for obtaining differentially private answer. We are in particular interested in a special class of monotonic queries that compute linear aggregation on a relation, and the relation is itself a function of the sensitive database.

\begin{definition}
    A linear query $q$ on sensitive database is a function $q=q_+\circ q_*$, where $q_*:\mathbb{D}\rightarrow\mathcal{P}(U\TUP)$ and $q_+:\mathcal{P}(U\TUP)\rightarrow\mathbb{R}$, such that $q_*$ transforms a database $D\in\mathbb{D}$ into a finite set of tuples~(e.g. by some relational algebra calculation), and $q_+$ is a linear function: $q_+(T)=\sum_{t\in T}q_+(t)$.
\end{definition}

Note that the output of $q_*$ must be finite, although the space $U\TUP$ can be infinite.

To ensure that a linear query $q$ is monotonic, we pose some limitations on the functions $q_*$ and $q_+$. First, we require that introducing a new participant into a sensitive database never results in removal of any tuple from the relation output by $q_*$. Second, we assume that $q_+$ is nonnegative.

\begin{definition}
    A linear query $q=q_+\circ q_*$ on sensitive database is monotonic if the following hold:
    \begin{itemize}
    \item
        $q_*(M_1(P_1))\subseteq q_*(M_2(P_2))$ for all $(P_1,M_1)\preceq(P_2,M_2)$
    \item
        $q_+(T)\geq 0$ for all finite $T\subseteq U\TUP$
    \end{itemize}
\end{definition}

If we want to answer a linear function $q_+$ that may yield negative output, we can decompose it into two nonnegative components and compute them individually: $q_+(t)=\max(0,q_+(t))-\max(0,-q_+(t))$.

Because we focus on a single query, where $q_*$ is fixed, we can construct a class of virtual sensitive databases $\Omega'=\{(P,M')\}$, such that each $(P,M)$ in $\Omega$ is mapped into a virtual one $(P,M')$ where $M'=q_*\circ M$. Then $M'(P)$ is a set of tuples and the query $q=q_+\circ q_*$ is just a linear function that computes $q_+(M'(P))$. The monotonicity of $q_*$ transmits to the monotonicity of $M'$. We call such $(P,M')$ a \emph{sensitive relation}.

\begin{definition}
    A sensitive relation $(P,M)$ is a sensitive database with $M:\mathcal{P}(P)\rightarrow\mathcal{P}(U\TUP)$, and $M(P)$ must be finite. A class $\Omega$ of sensitive relations is monotonic if $M_1(P_1)\subseteq M_2(P_2)$ for all $(P_1,M_1)\preceq (P_2,M_2)$ in $\Omega$.
\end{definition}

In this subsection and most parts of this paper, we study nonnegative linear queries for a monotonic class of sensitive relations.

To obtain a differentially private answer to a query $q$ on a relation $T=M(P)$, it should specify how the relation $T$ is affected by its contributors $P$. In particular, we want to know for each tuple in $T$ the condition of its presence if some participants may opt out. The definition of the function $M$ is too general to be efficiently handled in practice. Therefore, we propose to represent $M$ as a $c$-table or $K$-relation $R$, where each tuple $t$ is annotated with a positive Boolean expression $R(t)$ that specifies its condition of presence. Each variable $p$ in an expression indicates whether the participant $p\in P$ would contribute its data. A sensitive relation represented as a $K$-relation is called a sensitive $K$-relation, denoted by $(P,R)$.

For a query $q$, an algorithm may first transform the original sensitive database $(P,M)$ into a sensitive $K$-relation $(P,R)$ in a flexible way. For the correctness of the differentially private mechanism, however, the transformation should guarantee that for any neighboring sensitive databases the corresponding sensitive $K$-relations are also neighboring. The concept of \emph{neighboring} for sensitive $K$-relations is defined by as follows.

\begin{definition}
    \label{def:problem:krelationneighboring}
    Given an equivalence relation $\sim$ on $K$, two sensitive $K$-relations $(P_1,R_1)$ and $(P_2,R_2)$, where $P_2=P_1\cup\{p\}$, $p\notin P_1$, are neighboring if $R_1(t)\sim R_2(t)_{|p\rightarrow\FALSE}$ for all $t\in U\TUP$, where $R_2(t)_{|p\rightarrow\FALSE}$ denotes an operation that replaces all occurrences of the variable $p$ in $R_2(t)$ with constant $\FALSE$.
\end{definition}

An issue in the above definition is that it does not specify what kinds of Boolean expressions in $K$ are equivalent. A necessary condition for two expressions being equivalent is that they must yield the same truth-value for all valuation of variables. The way we write the expressions may, or may not matter, depending on the particular algorithms being used. For example, the inefficient mechanism presented in Sec.~\ref{sec:framework:inefficient} is independent of the form of expressions, so expressions that yield the same truth table are equivalent. On the other hand, the efficient mechanism presented in Sec.~\ref{sec:relaxation} relies on the way we write an expression. We will discuss this in Sec.~\ref{sec:relaxation}.

In Fig.~\ref{fig:problem:example} we present simple examples of $K$-relations that are produced by different queries to a graph. Fig.~\ref{fig:problem:example}(a) is a subgraph counting, while Fig.~\ref{fig:problem:example}(b) is a more complicated query.

\begin{figure*}[!t]
    \centering

    \mbox{
    \begin{tabular}[b]{c}
        \parbox{0.2\textwidth}{\includegraphics[width=0.2\textwidth]{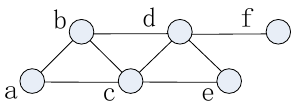}} \\
        social network (graph) \\
        \begin{tabular}{c}
            \begin{tabular}{|c|c|}
                \hline
                $t$ & $R(t)$ \\\hline
                $abc$ & $a\wedge b\wedge c$ \\\hline
                $bcd$ & $b\wedge c\wedge d$ \\\hline
                $cde$ & $c\wedge d\wedge e$ \\\hline
            \end{tabular} \\
            node differential privacy
        \end{tabular}
        \begin{tabular}{c}
            \begin{tabular}{|c|c|}
                \hline
                $t$ & $R(t)$ \\\hline
                $abc$ & $e_{ab}\wedge e_{ac}\wedge e_{bc}$ \\\hline
                $bcd$ & $e_{bc}\wedge e_{bd}\wedge e_{cd}$ \\\hline
                $cde$ & $e_{cd}\wedge e_{ce}\wedge e_{de}$ \\\hline
            \end{tabular}\\
            edge differential privacy
        \end{tabular} \\
        (a) How many triangles in a social network
    \end{tabular}

    \begin{tabular}[b]{c}
        \begin{tabular}{c}
            \begin{tabular}{|c|c|}
                \hline
                $t$ & $R(t)$ \\\hline
                $ab$ & $a\wedge b\wedge c$ \\\hline
                $ac$ & $a\wedge c\wedge b$ \\\hline
                $bc$ & $b\wedge c\wedge (a\vee d)$ \\\hline
                $bd$ & $b\wedge d\wedge c$ \\\hline
                $cd$ & $c\wedge d\wedge (b\vee e)$ \\\hline
                $ce$ & $c\wedge e\wedge d$ \\\hline
                $de$ & $d\wedge e\wedge c$ \\\hline
            \end{tabular} \\
            node differential privacy
        \end{tabular}
        \begin{tabular}{c}
            \begin{tabular}{|c|c|}
                \hline
                $t$ & $R(t)$ \\\hline
                $ab$ & $e_{ab}\wedge e_{ac}\wedge e_{bc}$ \\\hline
                $ac$ & $e_{ac}\wedge e_{ab}\wedge e_{bc}$ \\\hline
                $bc$ & $e_{bc}\wedge ((e_{ab}\wedge e_{ac})\vee (e_{bd}\wedge e_{cd}))$ \\\hline
                $bd$ & $e_{bc}\wedge e_{bd}\wedge e_{cd}$ \\\hline
                $cd$ & $e_{cd}\wedge ((e_{bc}\wedge e_{bd})\vee (e_{ce}\wedge e_{de}))$ \\\hline
                $ce$ & $e_{ce}\wedge e_{cd}\wedge e_{de}$ \\\hline
                $de$ & $e_{de}\wedge e_{cd}\wedge e_{ce}$ \\\hline
            \end{tabular}\\
            edge differential privacy
        \end{tabular} \\
        (b) How many pairs of friends that have a common friend
    \end{tabular}
    }
    \caption{Examples of $K$-relations}
    \label{fig:problem:example}
\end{figure*}

Finally, we introduce a variant of empirical sensitivity, which is relevant to the error bound of our mechanism.

\begin{definition}[Impact]
    For a sensitive $K$-relation $(P,R)$ and a participant $p\in P$, the impact of $p$ at $R$ is
    \begin{equation}
        \IMPACT(p,R)=\{t:R(t)\not\sim R(t)_{|p\rightarrow\FALSE}\}
    \end{equation}
\end{definition}

\begin{definition}[Universal Empirical Sensitivity]
    For a sensitive $K$-relation $(P,R)$, a participant $p\in P$ and a nonnegative linear query $q$, the universal empirical sensitivity of $q$ for a participant $p$ at $R$ is
    \begin{equation}
        \widetilde{US}_q(p,R)=\sum_{t\in\IMPACT(p,R)}q(t)
    \end{equation}
    For a sensitive $K$-relation $(P,R)$ and a nonnegative linear query $q$, the universal empirical sensitivity of $q$ at $(P,R)$ is
    \begin{equation}
        \widetilde{US}_q(P,R)=\max_{p\in P}\widetilde{US}_q(p,R)
    \end{equation}
\end{definition}


When $q(t)=1$ for all $t$, $\widetilde{US}_q(p,R)$ measures how many tuples in $R$ have $p$ appearing in their annotated expressions. The error bound of our mechanism presented in Sec.~\ref{sec:relaxation} is roughly proportional to the universal empirical sensitivity $\widetilde{US}_q$.

    \section{The Recursive Mechanism Framework}
\label{sec:framework}

In this section, we first present the framework of a novel differential privacy mechanism, \emph{recursive mechanism}, which can answer any monotonic queries on any sensitive databases. Then, we give a general but inefficient implementation of the mechanism.

\subsection{The Basic Framework}

Our mechanism is based on two special sequences, $H_0(P,M)\cdots H_{|P|}(P,M)$ and $G_0(P,M)\cdots G_{|P|}(P,M)$, as functions of the sensitive database $(P,M)$ in $\Omega$. We call $H$ a recursive sequence, which should satisfy the conditions given by the following definition.

\begin{definition}[Recursive Sequence]
    A sequence, $H_0(P,M)\ldots H_{|P|}(P,M)$, as a function on $\Omega$, is called a recursive sequence if the following conditions hold:
    \begin{itemize}
    \item
        $H_0(P,M)=0$ for all $(P,M)\in\Omega$
    \item
        (Recursive Monotonicity) $H_i(P_2,M_2)\leq H_i(P_1,M_1)\leq H_{i+1}(P_2,M_2)$ for all neighboring $(P_1,M_1)\preceq (P_2,M_2)$ in $\Omega$ and $0\leq i\leq |P_1|$
    \end{itemize}
\end{definition}

We call $G$ a bounding sequence of $H$, which is also a recursive sequence but satisfies some additional condition.

\begin{definition}[Bounding Sequence]
    For a recursive sequence $H$ and $g\geq 1$, a sequence, $G_0(P,M)\ldots G_{|P|}(P,M)$, as a function on $\Omega$, is called a $g$-bounding sequence of $H$, if the following conditions hold:
    \begin{itemize}
    \item
        $G$ is a recursive sequence
    \item
        $H_j(P,M)\leq H_i(P,M)+(|P|-i) G_k(P,M)$ for all $(P,M)\in\Omega$ and all $0\leq i\leq j\leq |P|$ and $k=|P|-\lfloor(|P|-j)/g\rfloor$
    \end{itemize}
    If $g=1$, we simply say $G$ is a bounding sequence of $H$.
\end{definition}

The framework of our mechanism consists of three steps
\begin{enumerate}
\item
    For a monotonic query $q$, we construct a recursive sequence $H$ and a $g$-bounding sequence $G$ of $H$ such that $H_{|P|}(P,M)=q(M(P))$ for all $(P,M)\in\Omega$.
\item
    Based on $G$, find a quantity $\Delta$ such that $\Delta$ approximates $G_{|P|}(P,M)$ or the empirical sensitivity of $q$, and $\ln\Delta$ has low global sensitivity, then we add multiplicative noise to $\Delta$, obtaining $\widehat{\Delta}$, which satisfies differential privacy.
\item
    Based on $H$, find a quantity $X$ such that $X$ approximates the true answer $H_{|P|}(P,M)$, and $X$ has global sensitivity $\widehat{\Delta}$, then we add Laplace noise to $X$, obtaining $\widehat{X}$, which satisfies differential privacy.
\end{enumerate}

The concrete construction of $H$ and $G$ are omitted here. We focus on Step 2 and 3 in this subsection. In the remainder of this paper, we will omit the argument $(P,M)$ when the context is clear.

For a sensitive database $(P,M)$ and parameters $\beta>0$ and $\theta>0$, we compute $\Delta$ as following
\begin{equation}
    \Delta=\min\{e^{i\beta}\theta:G_{|P|-i}\leq e^{i\beta}\theta\}
\end{equation}

We can observe several important properties of $\Delta$. In the sequel, all proofs of lemmas and theorems are moved to the appendix.

\begin{lemma}
    $GS_{\ln\Delta}\leq\beta$.
\end{lemma}

\begin{lemma}
    $\Delta\leq\max\{\theta,e^\beta G_{|P|}\}$.
\end{lemma}

\begin{lemma}
    $G_{|P|-\ln(\frac{\Delta}{\theta})/\beta}\leq\Delta$.
\end{lemma}

Because $\ln\Delta$ has low global sensitivity, we can add Laplace noise to $\ln\Delta$ to obtain a noisy version $\widehat\Delta$ that satisfies differential privacy. For parameter $\epsilon_1>0$ and $\mu>0$, we compute $\widehat{\Delta}=e^{\mu+Y}\Delta$, where $Y\sim\LAP(\beta/\epsilon_1)$. This finishes Step 2, and $\widehat{\Delta}$ has several properties.

\begin{lemma}
    The release of $\widehat{\Delta}$ satisfies $\epsilon_1$-differential privacy.
\end{lemma}

\begin{lemma}
    $\PR[\widehat{\Delta}>e^{\mu+c}\Delta]\leq\frac{1}{2}e^{-c\epsilon_1/\beta}$ for any $c>0$.
\end{lemma}

\begin{lemma}
    $\PR[\widehat{\Delta}<\Delta]\leq\frac{1}{2}e^{-\mu\epsilon_1/\beta}$.
\end{lemma}

In Step 3, we first find a quantity $X$ such that $X$ approximates $H_{|P|}(P,M)$ and $GS_X\leq\widehat{\Delta}$. We compute $X$ as
\begin{equation}
    X=\min\{H_i+(|P|-i)\widehat{\Delta}:0\leq i\leq |P|\}
\end{equation}

We have several properties of $X$.

\begin{lemma}
    For any fixed $\widehat{\Delta}\geq 0$, $GS_X\leq\widehat{\Delta}$.
\end{lemma}

\begin{lemma}
    If $\widehat{\Delta}\geq\Delta$, then $H_{|P|-g\ln(\frac{\Delta}{\theta})/\beta}\leq X\leq H_{|P|}$.
\end{lemma}

For parameter $\epsilon_2>0$, our mechanism releases $\widehat{X}=X+Y$, where $Y\sim\LAP(\widehat{\Delta}/\epsilon_2)$. We give the privacy and utility guarantees in the following theorem.

\begin{theorem}
    For parameters $\epsilon_1>0$, $\epsilon_2>0$, $\beta>0$, $\theta>0$ and $\mu>0$, recursive mechanism, as described above, satisfies $(\epsilon_1+\epsilon_2)$-differential privacy, and is $(e^{2\mu}\Delta^* c/\epsilon_2+
    g\lceil \ln(\frac{\Delta^*}{\theta})/\beta\rceil G_{|P|},
    e^{-\mu\epsilon_1/\beta}+
    e^{-c})$-accurate for any $c>0$, where $\Delta^*=\max\{\theta,e^\beta G_{|P|}\}$.
    If $\epsilon_1=\Theta(\epsilon)$, $\epsilon_2=\Theta(\epsilon)$, $\beta=\epsilon_1/k$, and $\theta$ and $\mu$ are constants, then the mechanism is $(O(k\ln(G_{|P|})G_{|P|}/\epsilon),2e^{-k\mu})$-accurate as $\epsilon\rightarrow 0$,$k\rightarrow\infty$ and $G_{|P|}\rightarrow\infty$.
\end{theorem}

The error bound of recursive mechanism is roughly proportional to $G_{|P|}$. Hence, the most important thing in a concrete implementation of recursive mechanism is to find sequences $H$ and $G$ with $G_{|P|}$ as small as possible.

\subsection{A General but Inefficient Implementation}
\label{sec:framework:inefficient}

Now we present a general but inefficient implementation of the recursive mechanism, which can answer any monotonic queries on sensitive databases. For a monotonic query $q$, we construct $H$ and $G$ as follows:

\begin{align}
    H_i(P,M)&=\min_{(P',M')\preceq(P,M),|P'|=i} q(M'(P'))\\
    G_i(P,M)&=\min_{(P',M')\preceq(P,M),|P'|=i} \widetilde{GS}_q(P',M')\\
\end{align}

Then we can show that the above $H$ and $G$ are what we want.

\begin{theorem}
    The sequence $H$ is a recursive sequence, and the sequence $G$ is a bounding sequence of $H$.
\end{theorem}

Because $G_{|P|}(P,M)=\widetilde{GS}_q(P,M)$, the error bound of recursive mechanism using these $H$ and $G$ is roughly proportional to the global empirical sensitivity of $q$. The main disadvantage of this implementation is the expensive computation cost for $H$ and $G$.

    \section{Efficient Recursive Mechanism}\label{sec:relaxation}

In this section, we present an efficient recursive mechanism, which takes polynomial computation cost and can answer linear queries on sensitive $K$-relations.

\subsection{Recursive Mechanism with Relaxation}

The central idea of the efficient recursive mechanism is relaxation, which introduces a mapping $\phi:K\rightarrow [0,1]^{[0,1]^P}$ that maps each Boolean expression in $K$ into a $[0,1]$-valued expression $\phi_k:[0,1]^P\rightarrow [0,1]$. The detail of $\phi$ will be discussed in the next subsection. Now, we first give some required properties of $\phi$.

For simplifying notations, we let $\TRUE=1$ and $\FALSE=0$. For $f:P\rightarrow [0,1]$, we define $|f|=\sum_p f(p)$. By $f\leq g$ we mean $f(p)\leq g(p)$ for all $p$. The mapping $\phi$ has the following properties.

\begin{description}
\item[Correctness]
    For any $k\in K$ and any Boolean assignment $f:P\rightarrow\{0,1\}$, $\phi_k(f)=k(f)$.
\item[Naturalness]
    For any $k\in K$, any real assignment $f:P\rightarrow[0,1]$ and any $p\in P$, if $f(p)=0$, then $\phi_k(f)=\phi_{k_{|p\rightarrow\FALSE}}(f)$, and if $f(p)=1$, then $\phi_k(f)=\phi_{k_{|p\rightarrow\TRUE}}(f)$.
\item[Monotonicity]
    For any $k\in K$ and any real assignments $f,g:P\rightarrow[0,1]$, if $f\leq g$, then $\phi_k(f)\leq\phi_k(g)$.
\item[Convexity]
    For any $k\in K$, $\phi_k$ is a convex function.
\item[Truncated Linearity]
    Define $\psi(x)=\min(1,x)$ and $\phi^*_k(f)=1-\phi_k(1-\psi\circ f)$. For any $k\in K$, $f:P\rightarrow[0,1]$ and $c\geq 1$, $\phi^*_k(c f)=\min(1,c\phi^*_k(f))$
\end{description}

Then, we introduce the notion of \emph{equivalence} --- two Boolean expressions in $K$ are equivalent if their relaxed functions under $\phi$ are the same. This completes Definition~\ref{def:problem:krelationneighboring} for neighboring sensitive $K$-relations.

\begin{definition}[Equivalence]
    For any $k_1,k_2\in K$, $k_1$ and $k_2$ are equivalent, denoted by $k_1\sim k_2$, if $\phi_{k_1}=\phi_{k_2}$.
\end{definition}

Equivalence of two expressions implies that they yield the same truth table. But expressions that yield the same truth table are not necessarily equivalent. We will explain this in the next subsection.

Provided a nonnegative linear query $q:U\TUP\rightarrow\mathbb{R}$ and mapping $\phi$, we construct the recursive sequence $H$ as

\begin{equation}
    H_i(P,R)=\min_{f\in[0,1]^P,|f|=i}\sum_t q(t)\phi_{R(t)}(f) \label{eq:relaxation:h}
\end{equation}
Note that the sum is finite since $R$ has finite support.

\begin{theorem}
    The sequence $H$ is a recursive sequence, and $H_{|P|}(P,R)=q(\SUPP(R))$.
\end{theorem}

To construct the bounding sequence of $H$, we also require that an auxiliary quantity $S_{k,p}$ is provided for each $k\in K$ and $p\in P$, which bounds the maximum change of $\phi_k(f)$ caused by a small change to $f(p)$. Formally, for all $f,g\in[0,1]^P$, if $f\leq g$, and $f(p')=g(p')$ for all $p'\in P-\{p\}$, then
\begin{equation}
    \phi_k(g)-\phi_k(f)\leq (g(p)-f(p))S_{k,p}
\end{equation}
$S_{k,p}$ can be seen as the upper bound of the partial derivative of $\phi_k$ w.r.t. $p$. We call $S_{k,p}$ the $\phi$-sensitivity of the expression $k$ for $p$. We can observe the following fact.

\begin{lemma}
    For any $f\leq g$ in $[0,1]^P$, and any $k\in K$,
    \begin{equation}
        \phi_k(g)-\phi_k(f)\leq\sum_p (g(p)-f(p))S_{k,p}\leq |g-f|\max_p S_{k,p}
    \end{equation}
\end{lemma}

Assuming that all $\phi$-sensitivities $S_{k,p}$ are known, we construct a $2$-bounding sequence $G$ of $H$ as

\begin{equation}
    G_i(P,R)=2\min_{f\in[0,1]^P,|f|=i}\max_{p\in P}\sum_t q(t)\phi_{R(t)}(f)S_{R(t),p}
\end{equation}

\begin{theorem}
    The sequence $G$ is a $2$-bounding sequence of $H$.
\end{theorem}

\subsection{The Mapping $\phi$}

Here we discuss the mapping $\phi$, the issues about annotation of Boolean expressions, and the utility guarantee of the recursive mechanism. For an expression $k$, we define $\phi_k$ in a recursive way, as follows:
\begin{itemize}
\item
    $\phi_{\FALSE}(f)=0$ and $\phi_{\TRUE}(f)=1$ for all $f$
\item
    $\phi_p(f)=f(p)$ for all $p\in P$
\item
    $\phi_{x\wedge y}(f)=\max\{0,\phi_x(f)+\phi_y(f)-1\}$ and $\phi_{x\vee y}(f)=\max\{\phi_x(f),\phi_y(f)\}$ for all expressions $x$ and $y$
\end{itemize}

It can be shown that the above $\phi$ is just what we need.

\begin{theorem}
    The mapping $\phi$, defined above, have the desired properties of correctness, naturalness, monotonicity, convexity, and truncated linearity.
\end{theorem}

The output of mapping $\phi$ is invariant under certain kinds of transformations of the input expressions.
\begin{description}
\item[Identity]
    $\phi_{x\wedge\TRUE}=\phi_x$, $\phi_{x\vee\FALSE}=\phi_x$
\item[Annihilator]
    $\phi_{x\wedge\FALSE}=\phi_{\FALSE}$, $\phi_{x\vee\TRUE}=\phi_{\TRUE}$
\item[Associativity]
    $\phi_{x\wedge(y\wedge z)}=\phi_{(x\wedge y)\wedge z}$, $\phi_{x\vee(y\vee z)}=\phi_{(x\vee y)\vee z}$
\item[Distributivity of $\wedge$ over $\vee$]
    $\phi_{x\wedge(y\vee z)}=\phi_{(x\wedge y)\vee(x\wedge z)}$
\end{description}
Two expressions are equivalent if one can be obtained from another via a series of above transformations. Because $\phi$ is defined recursively, the above transformations can be applied to any place of an expression $k$ without changing $\phi_k$.

Before invoking our mechanism, one needs to first generate a sensitive $K$-relation from the sensitive database and then issue a monotonic query. To satisfy differential privacy, it is important to ensure that for any neighboring sensitive databases the resulting sensitive $K$-relations are still neighboring, according to Definition~\ref{def:problem:krelationneighboring}. Hence, when we annotate tuples with expressions that specify their conditions of presence, we should take care of the way we write the expressions. Specifically, if a tuple $t$ is annotated with expression $k$, then we should ensure that when any participant $p$ opts out, the new expression $k'$ annotated with $t$ can be obtained from $k_{|p\rightarrow\FALSE}$ via a series of invariant transformations. If this is guaranteed, then we say the annotation is safe. Fortunately, safe annotation is often easy to achieve. For positive relational algebra queries, the annotation provided in Sec.~\ref{sec:pre:krelation} is always safe. Moreover, two expressions in disjunctive normal form are equivalent if and only if they produce the same truth table. Therefore, if we always expand all expressions into disjunctive normal form, then the annotation is always safe.

The $\phi$-sensitivities $S_{k,p}$, which bound the partial derivative of $\phi_k$ w.r.t. $p$, are also computed in a recursive way
\begin{itemize}
\item
    $S_{\TRUE,p}=S_{\FALSE,p}=0$ and $S_{p,p}=1$
\item
    $S_{x\wedge y,p}=S_{x,p}+S_{y,p}$ and $S_{x\vee y,p}=\max\{S_{x,p},S_{y,p}\}$
\end{itemize}

We can observe several properties of $\phi$-sensitivities: 1) $S_{k,p}$ is not greater than the number of occurrences of $p$ in expression $k$; 2) $S_{k,p}$ is at most one plus the the number of occurrences of $\wedge$ in $k$; 3) if $k$ is written in disjunctive normal form (e.g., the case of subgraph counting), then $S_{k,p}\leq 1$; 4) for positive relational algebra query, if each tuple in the input tables is associated with at most one participant, and we use the approach described in Sec.~\ref{sec:pre:krelation} to annotate the tuples in the output table with expressions, then $S_{k,p}$ is at most one plus the number of operations in the positive relational algebra query. In Fig.~\ref{fig:relaxation:phisensi} we present several examples of $\phi$-sensitivity.

If we take the maximum $S$ of $S_{k,p}$ over all $k\in\{R(t)\}$ and $p\in P$, then we can find that $G_{|P|}(P,R)\leq 2S\cdot\widetilde{US}_q(P,R)$. Hence, we conclude that the error bound of our mechanism is roughly proportional to $S$ times the universal empirical sensitivity of $q$. In general, $S$ is linear in the length of the positive relational algebra query. If all expressions are converted to disjunctive normal form, then $S$ is just a constant 1. In particular, for subgraph counting we have $\widetilde{US}_q=\widetilde{GS}_q=\widetilde{LS}_q$. Thus the error bound is roughly proportional to the local empirical sensitivity of $q$.

\begin{figure}[!t]
    \begin{tabular}{|c|c|}
        \hline
        expression $k$ & $\phi$-sensitivities of $k$ \\\hline
        $a\wedge b\wedge c$ & $S_{k,a}=S_{k,b}=S_{k,c}=1$ \\\hline
        $(a\vee b)\wedge(a\vee c)\wedge(b\vee d)$ & $S_{k,a}=S_{k,b}=2$, $S_{k,c}=S_{k,d}=1$ \\\hline
        $(a\wedge b)\vee(a\wedge c)\vee(b\wedge d)$ & $S_{k,a}=S_{k,b}=S_{k,c}=S_{k,d}=1$ \\\hline
    \end{tabular}
    \caption{Examples of $\phi$-sensitivities}
    \label{fig:relaxation:phisensi}
\end{figure}

\subsection{Computation Cost}

Note that the computation for each $H_i$ and $G_i$ can be encoded into a linear program with $O(L)$ variables, where $L$ denotes the total length of all annotated expressions $R(t)$ for $t\in\SUPP(R)$. Therefore, our mechanism can run in polynomial time.

A simple algorithm that computes all $H_i$ and $G_i$ will need to solve $O(|P|)$ linear programs. We can improve this by utilizing the monotonicity of $G$ and the convexity of $H$.

\begin{lemma}
    (Convexity of $H$) $H_{i+1}-H_i\leq H_{i+2}-H_{i+1}$ for all $0\leq i\leq |P|-2$.
\end{lemma}

Let $j=\arg\min\{e^{j\beta}\theta:G_{|P|-j}\leq e^{j\beta}\theta\}$, then $\Delta=e^{j\beta}\theta$. We can observe that $j=\ln(\frac{\Delta}{\theta})/\beta\leq 1+\ln(\frac{G_{|P|}}{\theta})/\beta$. Hence, $\Delta$ can be computed with access to the last $O(\ln(G_{|P|})/\beta)$ entries of $G$. Furthermore, because $G_{|P|-j}-e^{j\beta}\theta$ is monotonously decreasing, we can use binary search to find $j$, with access to only $O(\ln(\ln(G_{|P|})/\beta))$ entries of $G$.

Given $\Delta$ and $\widetilde{\Delta}$, we then compute $X=H_i+(|P|-i)\widehat{\Delta}$, where $i=\arg\min\{H_i+(|P|-i)\widehat{\Delta}:0\leq i\leq |P|\}$. To do this, we compute
\begin{align}
    i'={\arg\min}_{i'\in [0,|P|]} H_{i'}+(|P|-i')\widehat{\Delta}
\end{align}

In the above formula, the range of $i'$ is a real interval rather than an integer, and the definition of $H_{i'}$ is the same as Eq.~\ref{eq:relaxation:h}. So $i'$ can be computed by solving a linear program. Due to the convexity of $H$, we also know that $\lfloor i'\rfloor\leq i\leq \lceil i'\rceil$. Hence, $i$ can then be computed with access to only two entries of $H$.

\begin{theorem}
    Efficient recursive mechanism can run in $O(\ln(\ln(G_{|P|})/\beta) T(L))$ time, where $T(L)$ denotes time needed to solve a linear program with $O(L)$ variables, and $L$ denotes the total length of all annotated expressions $R(t)$ for $t\in\SUPP(R)$.
\end{theorem}

    \section{Experimental Evaluation}

\begin{figure*}[!t]
    \centering
    \includegraphics[height=10pt]{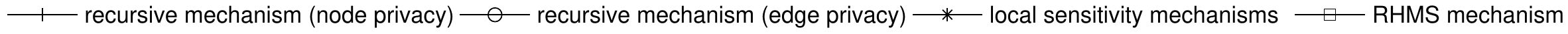}
    \begin{tabular}{m{8pt}@{}m{0.7\textwidth}}
        \rotatebox{90}{median relative error} & \includegraphics[width=0.7\textwidth]{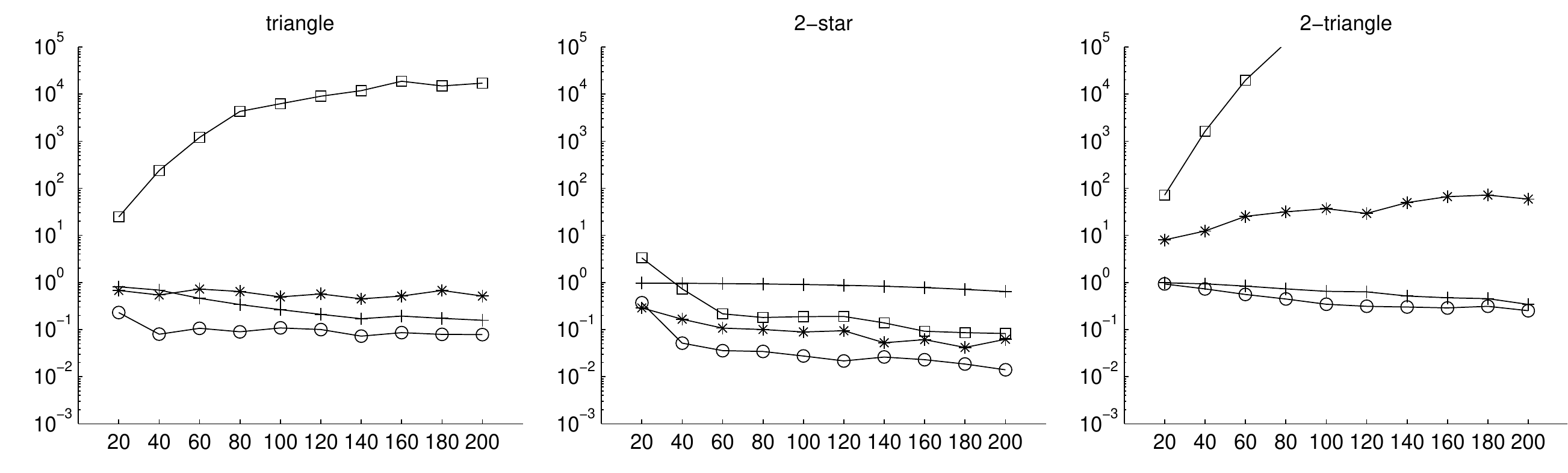} \\
        & \centering number of nodes
    \end{tabular}\\
    (a) Comparison on graphs with various number of nodes. $\AVGDEG=10$.
    \begin{tabular}{m{8pt}@{}m{0.7\textwidth}}
        \rotatebox{90}{median relative error} & \includegraphics[width=0.7\textwidth]{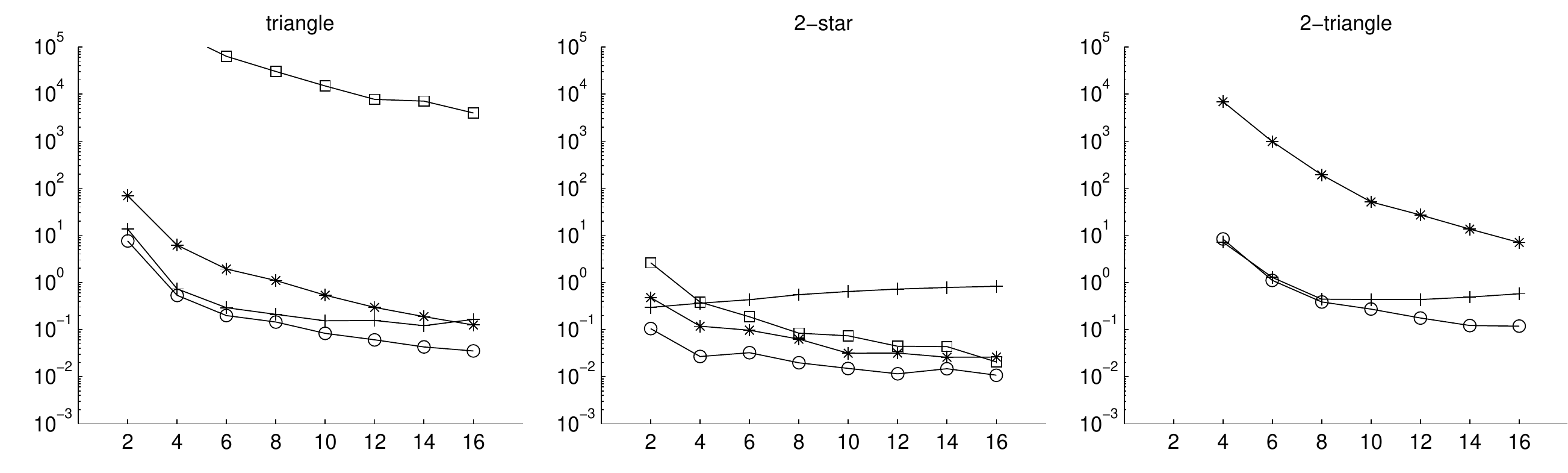} \\
        & \centering average degree
    \end{tabular}\\
    (b) Comparison on graphs with various average degrees. $|V|=200$.
    \begin{tabular}{m{8pt}@{}m{0.7\textwidth}}
        \rotatebox{90}{median relative error} & \includegraphics[width=0.7\textwidth]{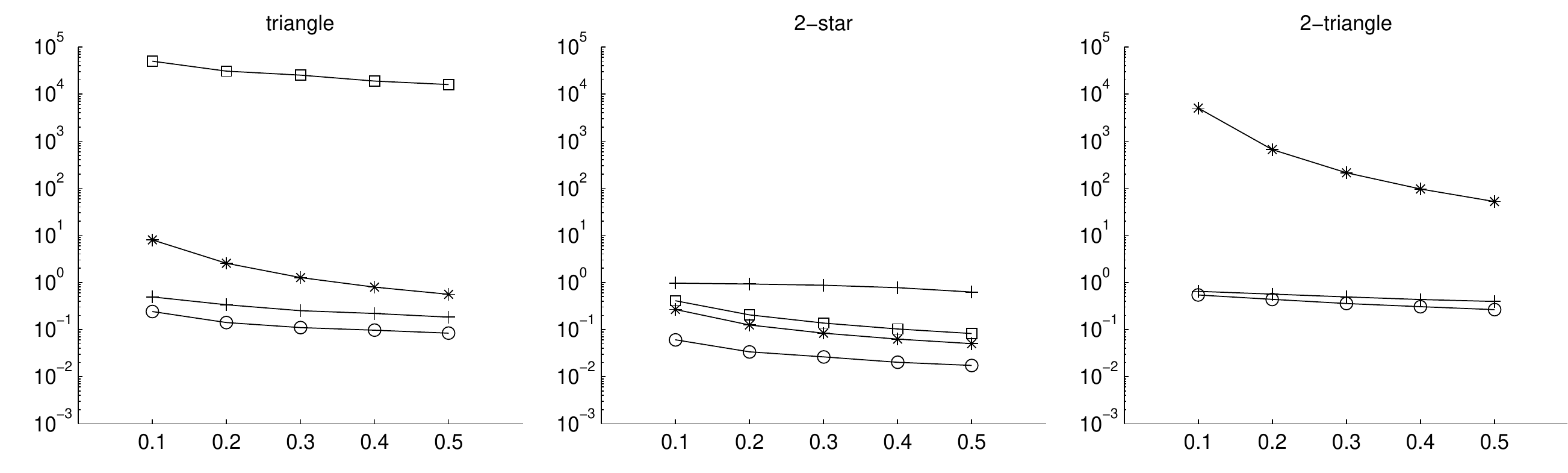} \\
        & \centering $\epsilon$
    \end{tabular}\\
    (c) Comparison with various $\epsilon$. $|V|=200$ and $\AVGDEG=10$.
    \caption{Comparing accuracy of different mechanisms in various settings.}
    \label{fig:exp:subgraph:compare}
\end{figure*}

In this section, we empirically evaluate the performance of our mechanism. We first compare our mechanism with existing mechanisms for answering subgraph counting queries, then we use our mechanism to process more general $K$-relations.

For each experiment, we generate several different graphs (or $K$-relations) by random, and for each graph we run every mechanism many times to obtain a series of answers. We measure the accuracy of mechanisms by median relative error, that is, the median of the ratios between the absolute errors and the true answers. This measure of accuracy is consistent with the work~\cite{DBLP:journals/pvldb/KarwaRSY11}.

\subsection{Subgraph Counting}

For subgraph counting, we compare the accuracy of our mechanism with the following existing mechanisms:

\emph{Local sensitivity mechanisms} include the triangle algorithm of \cite{DBLP:conf/stoc/NissimRS07}, the $k$-star algorithm and the $k$-triangle mechanism of \cite{DBLP:journals/pvldb/KarwaRSY11}. All algorithms are based on the local sensitivity of the query. The $k$-triangle algorithm achieves only $(\epsilon,\delta)$-differential privacy, while the others can achieve $\epsilon$-differential privacy.

\emph{RHMS mechanism} of \cite{DBLP:conf/pods/RastogiHMS09} can process subgraph counting for any connected subgraphs. It achieves  only $(\epsilon,\gamma)$-adversarial privacy for a specific class of adversaries.

We set $\epsilon=0.5$ and $\delta=\gamma=0.1$, which follows the parameter setting of \cite{DBLP:journals/pvldb/KarwaRSY11}.\footnote{It is widely believed by researchers that to provide useful privacy guarantee $\delta$ should be a negligible function of database size ~\cite{DBLP:conf/stoc/NissimRS07,DBLP:conf/stoc/DworkL09,DBLP:conf/focs/DworkRV10,DBLP:conf/stoc/RothR10} (i.e., $\delta$ is asymptotically smaller than any inverse polynomial: $\delta=1/|P|^{\omega(1)}$). However, the $k$-triangle algorithm~\cite{DBLP:journals/pvldb/KarwaRSY11} yields too noisy answers for such small $\delta$.} Our mechanism can achieve $\epsilon$-differential privacy, which is much stronger than the corresponding $(\epsilon,\delta)$-differential privacy and $(\epsilon,\gamma)$-adversarial privacy. We test two versions of our mechanism, one provides node privacy, and the other provides edge privacy. Because node privacy requires that the released answer must be insensitive to the change of one node and all of its incident edges, it needs to introduce noise of much greater magnitude into the answer. Note that all other mechanisms in comparison can only provide edge privacy. For our mechanism, we simply set $\theta=1$, $\beta=\epsilon/5$ and $\mu=0.5$, and we set $\mu=1$ for node differential privacy.

We first perform experiments on synthetic graphs that are generated by random. We generate graphs with various numbers of nodes and average degree $\AVGDEG$. Each edge in the graph appears independently with probability $\AVGDEG/(|V|-1)$. The experimental results are presented in Fig.~\ref{fig:exp:subgraph:compare}.

\begin{figure}[!t]
    \centering
    \includegraphics[width=1.1\columnwidth]{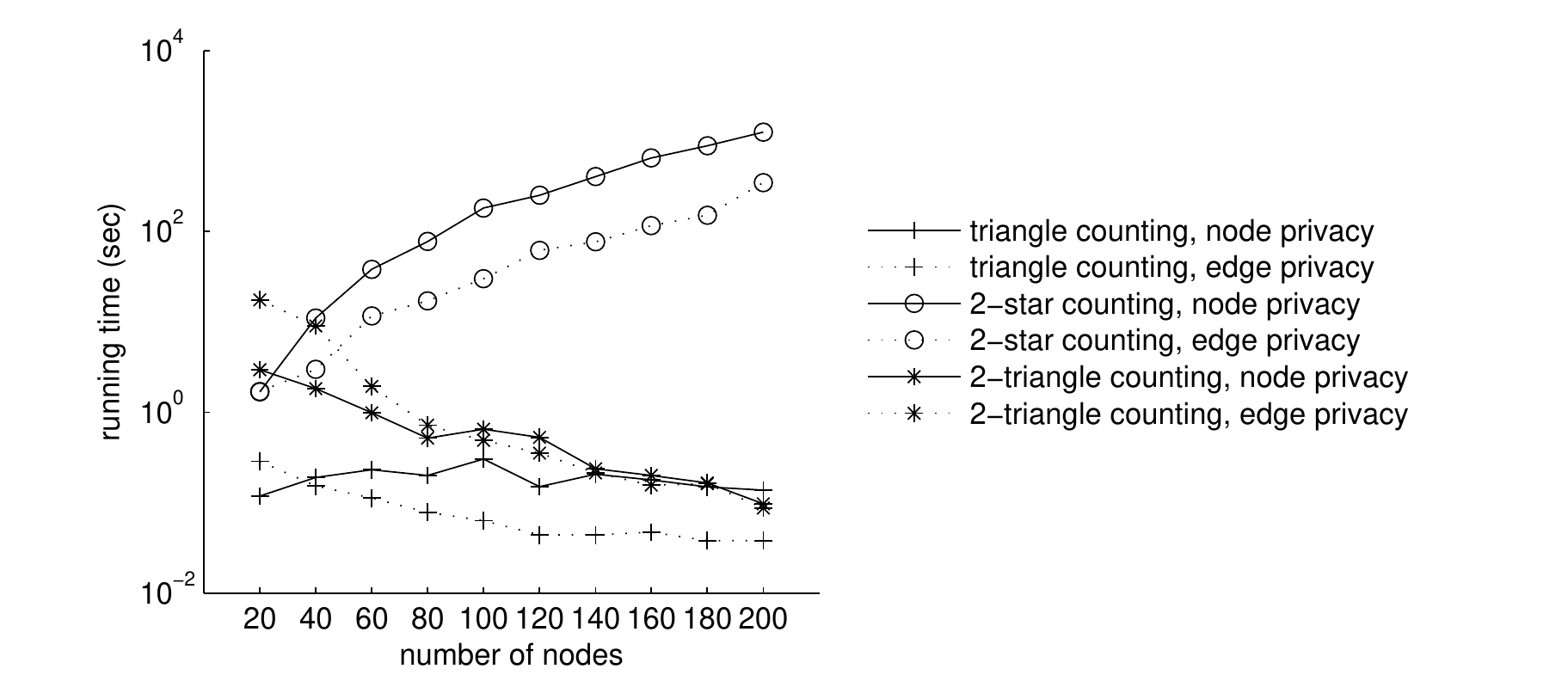}
    \caption{Running time of recursive mechanism, $\AVGDEG=10$.}
    \label{fig:exp:subgraph:time_n}
\end{figure}

\begin{figure}[!t]
    \centering
    \scalebox{0.7}{
    \begin{tabular}{|c|c|c|c|c|c|c|c|}
        \hline
        & netscience & power & 1138\_bus & bcspwr10 & gemat12 & ca-GrQc & ca-HepTh \\\hline
        $|V|$ & 1589 & 4941 & 1138 & 5300 & 4929 & 5242 & 9877 \\\hline
        $|E|$ & 2742 & 6594 & 2596 & 13571 & 33111 & 14496 & 25998 \\\hline
        \multicolumn{8}{c}{total number of triangles} \\\hline
        & 3764 & 651 & 128 & 721 & 592 & 48260 & 28339 \\\hline
        \multicolumn{8}{c}{running time of recursive mechanism in seconds (node privacy)} \\\hline
        & 8.924 & 0.468 & 0.078 & 0.655 & 0.640 & 3940.552 & 788.160 \\\hline
        \multicolumn{8}{c}{running time of recursive mechanism in seconds (edge privacy)} \\\hline
        & 19.188 & 0.374 & 0.063 & 0.406 & 0.483 & 54236.462 & 2104.818 \\\hline
    \end{tabular}}
    \caption{Sizes of real graphs and running time of our mechanism for triangle counting.}
    \includegraphics[width=\columnwidth]{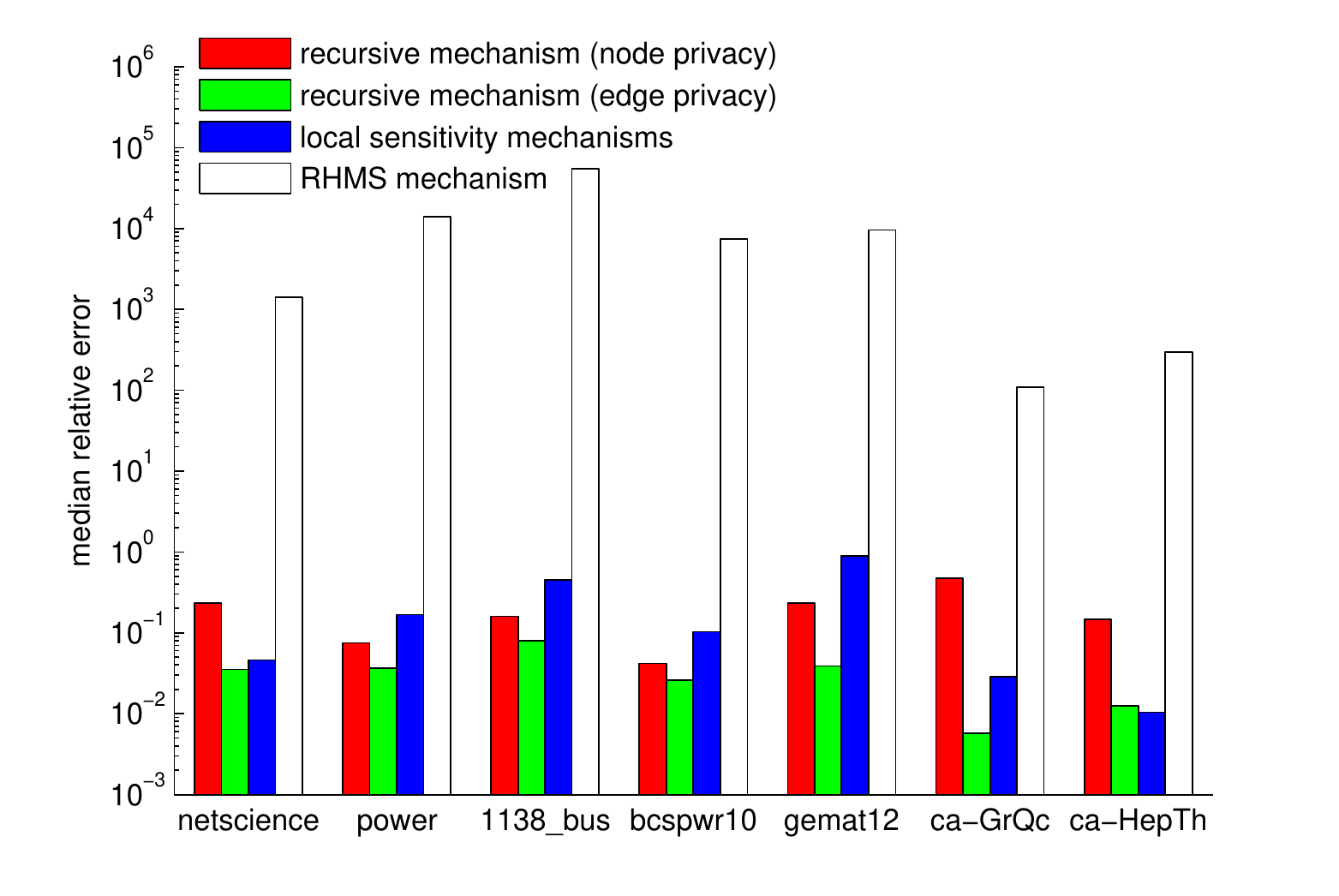}
    \label{fig:exp:subgraph:info_file}
    \vspace{-1cm}
    \caption{Comparing accuracy of different mechanisms for triangle counting on different graphs.}
    \label{fig:exp:subgraph:error_file}
\end{figure}

It can be observed that RHMS mechanism does not yield meaningful answers for triangle counting and $2$-triangle counting. This is because that its error bound grows exponentially with the number of edges in the subgraph. In some experiments, the relative errors of RHMS mechanism are extremely high and the curves do not show in the figures. Moreover, the errors of local sensitivity mechanisms are also too high to be useful for triangle counting and $2$-triangle counting when the graph is very sparse, because the smooth upper bound of local sensitivity is often high (relative to the true answer) for triangle counting on sparse graphs.

Our mechanism, when providing edge privacy (the same as other compared mechanisms), always yield the most accurate answers. When providing node privacy, our mechanism has high relative error for $2$-star counting and $2$-triangle counting, this is because the change of one node can affect a large number of $2$-stars and $2$-triangles in the graph. Nonetheless, the relative error of our mechanism decreases while the size of graph grows.

In Fig.~\ref{fig:exp:subgraph:time_n}, we present the running time of our mechanism. Because each matched subgraph found in the whole graph contributes a tuple into the $K$-relation, the computation cost of our mechanism grows polynomially with the true answer. Since the average degree is fixed, the number of triangles and $2$-triangles often decreases when the graph enlarges, hence our mechanism runs faster for large sparse random graphs. On the other hand, the number of $2$-stars is roughly proportional to the number of nodes, so the running time of our mechanism grows with the graph size\footnote{When the degrees of nodes are large, the number of $k$-stars and $k$-triangles can grow exponentially with $k$. One may think that our mechanism has exponential computation cost in this situation. Actually, the algorithm can be improved by a clever construction of $K$-relation, such that the size of $K$-relation is asymptotically independent of $k$. Due to limitation of space, we cannot present the details in this paper.}.

We also evaluate the mechanisms on several real datasets\footnote{Available at http://www.cise.ufl.edu/research/sparse/matrices/}. Experimental results are shown in Fig.~\ref{fig:exp:subgraph:info_file} and \ref{fig:exp:subgraph:error_file}. We can see that our mechanism are often superior to the other mechanisms. This validates the practical usage of our mechanism.

\begin{figure}[!t]
    \centering
    \mbox{
    \hspace{-.2cm}\includegraphics[width=.6\columnwidth]{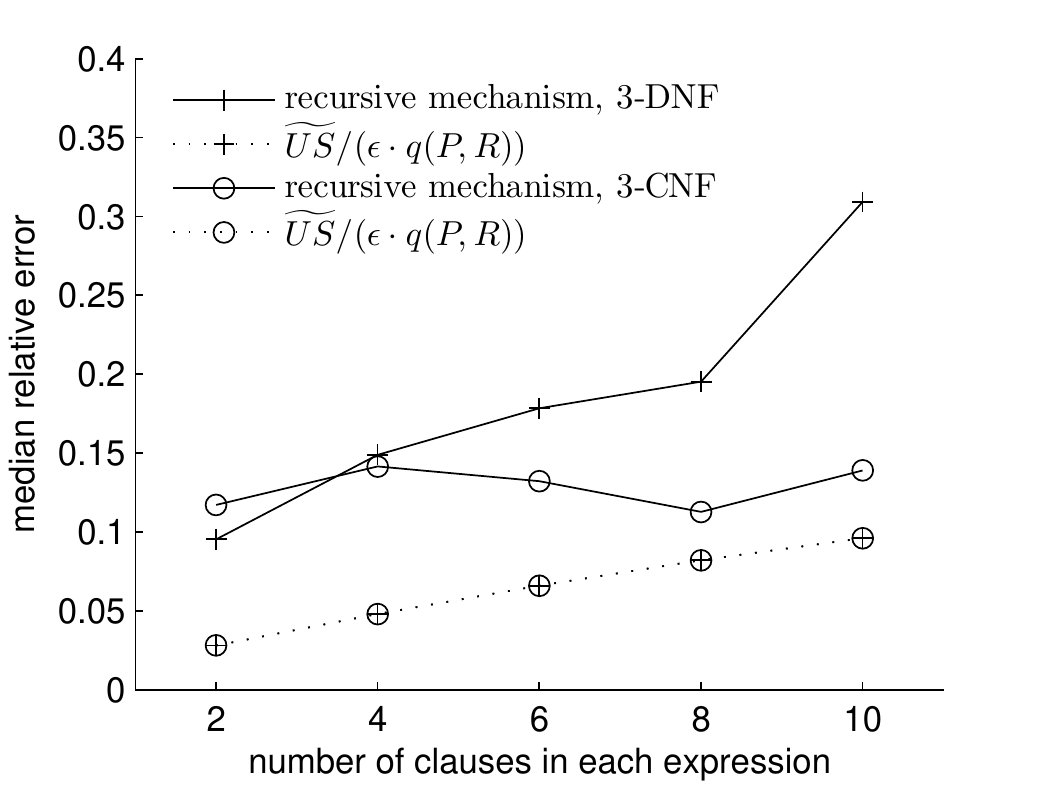}%
    \hspace{-.6cm}\includegraphics[width=.6\columnwidth]{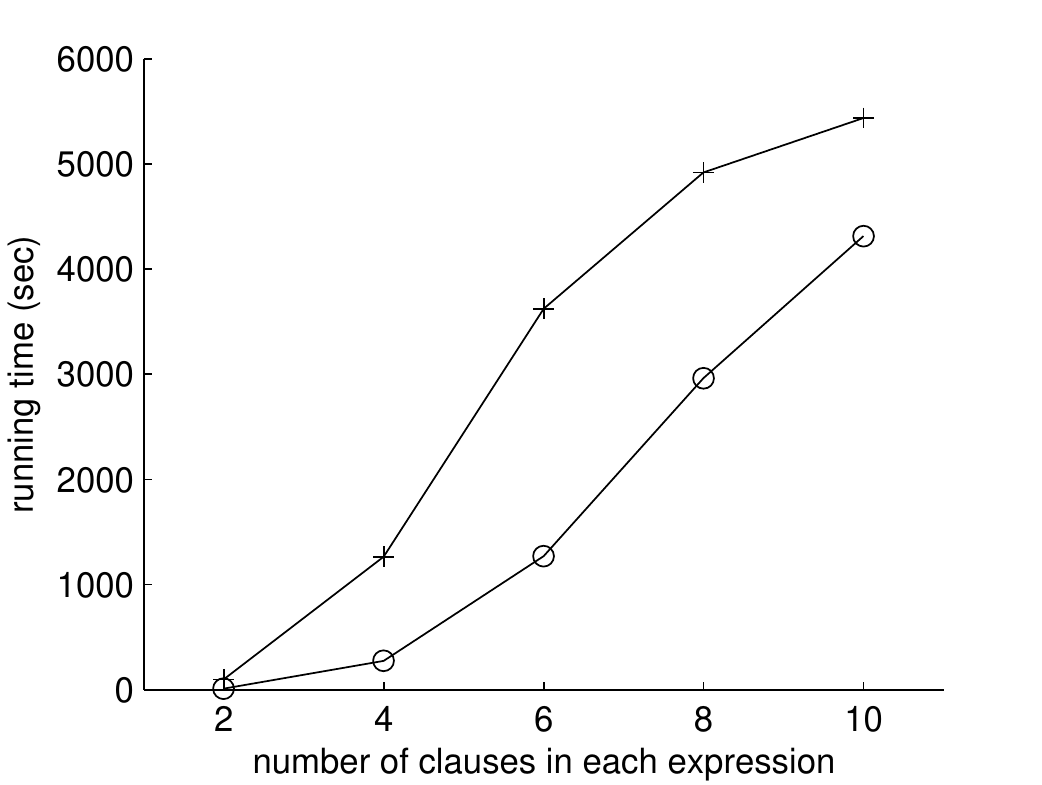}
    }
    \caption{Evaluating recursive mechanism on $K$-relations with various length of expressions, $|\SUPP(R)|=1000$.}
    \label{fig:exp:krelation:k1}

    \mbox{
    \hspace{-.2cm}\includegraphics[width=.6\columnwidth]{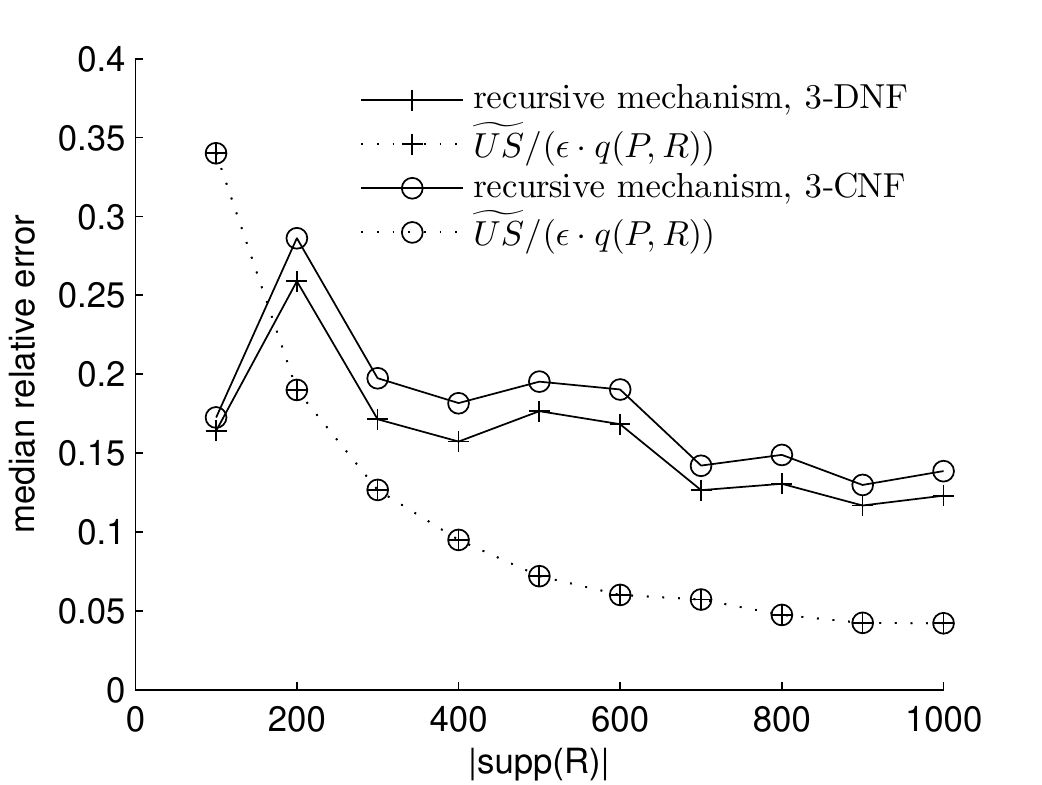}%
    \hspace{-.6cm}\includegraphics[width=.6\columnwidth]{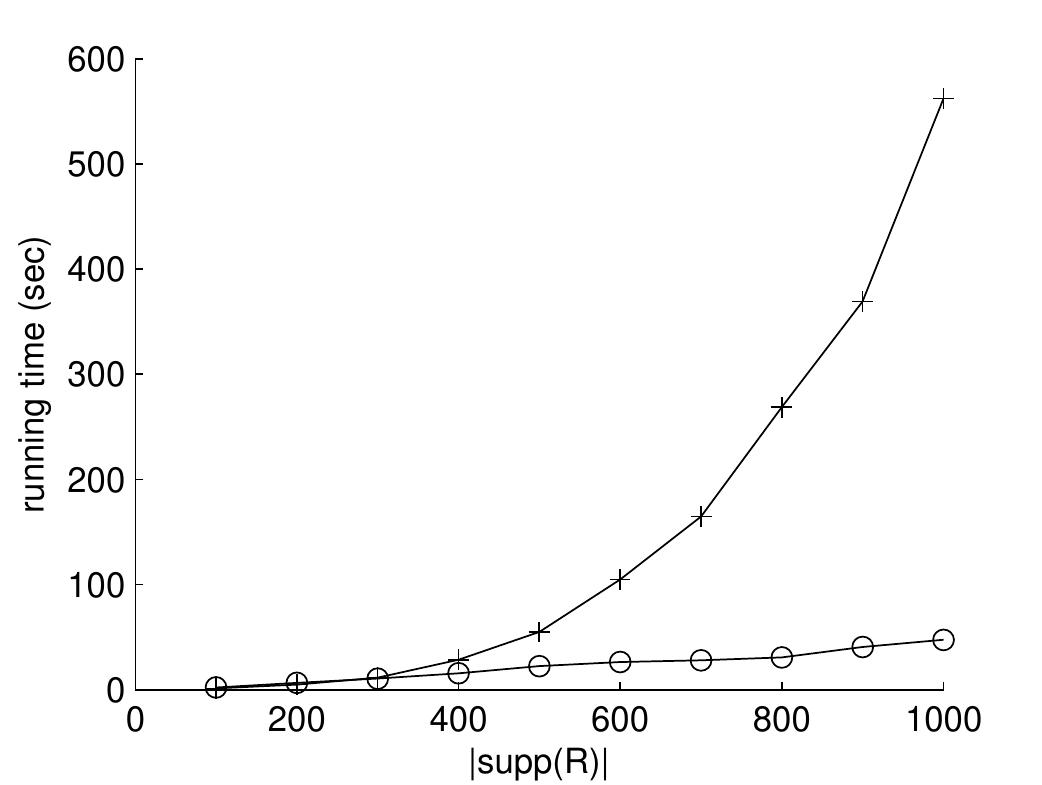}
    }
    \caption{Evaluating recursive mechanism on $K$-relations of various sizes, each expression has 3 clauses.}
    \label{fig:exp:krelation:size}
\end{figure}


\subsection{Processing $K$-Relations}

Finally, we evaluate the performance of our mechanism for processing more general queries. Because there are many different kinds of positive relational algebra queries, we directly generate $K$-relations that could be produced by some relational queries. In particular, we consider two kinds of $K$-relations: $K$-relations in which every tuple is annotated with a $3$-DNF Boolean expression, and $K$-relations in which every tuple is annotated with a $3$-CNF Boolean expression. A $3$-DNF $K$-relation can be produced by a union of many join results, and a $3$-CNF $K$-relation can be produced by a join of many unions of tables. We simply generate all expressions by random, but ensure that all annotated expressions have the same length. We also make $|P|$, the total number of variables, equal to $|\SUPP(R)|$, the size of the $K$-relation. We let $q(t)=1$, that is, the true answer is just $|\SUPP(R)|$. The performance of our mechanism is shown in Fig.~\ref{fig:exp:krelation:k1} and \ref{fig:exp:krelation:size}. We do not present experimental results for different kinds of $q(t)$ because the curves are almost the same.

The dotted curves in the figures denote the relative error if the absolute error exactly matches $\widetilde{US}_q/\epsilon$, where $\widetilde{US}_q$ is the maximum number of tuples that have at least one common participant appearing in their annotated expressions. The error of our mechanism is nearly linear in $\widetilde{US}_q/\epsilon$, as shown in the figures. The empirical sensitivity $\widetilde{US}_q$ is insensitive to the increase of the number of participants and the number of tuples in $R$. Hence the relative error of our mechanism can gradually decrease if more data are available. In terms of computation cost, the running time of our mechanism grows polynomially with $|\SUPP(R)|$ and the length of expressions.

    \section{Related Work}
\label{sec:relatedwork}

Since differential privacy was introduced~\cite{DBLP:conf/tcc/DworkMNS06}, it has gained considerable attention, and many techniques were developed for private data analysis. Dwork {\it et al.}~\cite{DBLP:conf/tcc/DworkMNS06} showed that differential privacy can be achieved if we calibrate the noise to the global sensitivity of the query, and proposed the Laplace mechanism. The noise yielded by Laplace mechanism is independent of the database instance. Due to simplicity and wide applicability of Laplace mechanism, many succeeding work for various query tasks were built upon Laplace mechanism. These include the relevant work~\cite{DBLP:conf/sigmod/McSherry09,Palamidessi12} that studied relational algebra queries under differential privacy.

Laplace mechanism fails to provide useful answers for queries that have large global sensitivity. Nissim {\it et al.}~\cite{DBLP:conf/stoc/NissimRS07} introduced the notion of local sensitivity, and proposed to calibrate the noise to a smooth upper bound of the local sensitivity. This leads to the idea of instance-dependent noise. They also gave algorithms for computing smooth sensitivity of triangle count as well as some other statistics in a variety of domains. However, there are no general way to compute the smooth sensitivity of a given query.

Inspired by the work \cite{DBLP:conf/stoc/NissimRS07}, Karwa {\it et al.}~\cite{DBLP:journals/pvldb/KarwaRSY11} studied the problem of $k$-star counting and $k$-triangle counting, which were based on the local sensitivity of the query. Rastogi {\it et al.}~\cite{DBLP:conf/pods/RastogiHMS09} addressed counting of general subgraphs, which achieves utility better than global sensitivity based mechanism by relaxing the privacy guarantee. These approaches only provide edge privacy guarantee.

There were also extensive studies on privacy in graph data beyond the scope of differential privacy, but most do not provide qualitative privacy and utility guarantee. Readers can refer to the survey~\cite{DBLP:journals/sigkdd/ZhouPL08} for techniques that are based on $k$-anonymity. 
    \section{Conclusion}

In this paper, we have presented a novel differentially private mechanism for releasing an approximation to a linear statistic of a table output by some positive relational algebra query to a database. It turns subgraph counting as a special case, and can provide guarantee of either node differential privacy or edge differential privacy. Empirical evaluation shows that our mechanism can return more accurate answer than existing algorithms for subgraph counting, while achieving the same or even stronger privacy guarantee.

    \bibliographystyle{abbrv}
    \bibliography{references}

    \appendix

\section{Proofs}

\setcounter{lemma}{0}
\setcounter{theorem}{0}

\begin{lemma}
    $GS_{\ln\Delta}\leq\beta$.
\end{lemma}

\begin{proof}[Sketch]
    For all neighboring $(P_1,M_1)\preceq(P_2,M_2)$, let $i=\arg\min_i\{e^{i\beta}\theta:G_{|P_1|-i}(P_1,M_1)\leq e^{i\beta}\theta\}$ and $j=\arg\min_j\{e^{j\beta}\theta:G_{|P_2|-j}(P_2,M_2)\leq e^{j\beta}\theta\}$, we show that $i\leq j\leq i+1$. Note that $G$ is a recursive sequence. Because $e^{(i-1)\beta}\theta<G_{|P_1|-(i-1)}(P_1,M_1)\leq G_{|P_2|-(i-1)}(P_2,M_2)$, we have $j\geq i$. Similarly, because $e^{(i+1)\beta}\theta\geq G_{|P_1|-i}(P_1,M_1)\geq G_{|P_2|-(i+1)}(P_2,M_2)$, we have $j\leq i+1$.
\end{proof}

\begin{lemma}
    $\Delta\leq\max\{\theta,e^\beta G_{|P|}\}$.
\end{lemma}

\begin{proof}[Sketch]
     Suppose $\Delta=e^{i\beta}\theta$. If $i=0$, then $\Delta=\theta$, otherwise, $\Delta=e^\beta e^{(i-1)\beta}\theta<e^\beta G_{|P|-i}\leq e^\beta G_{|P|}$.
\end{proof}

\begin{lemma}
    $G_{|P|-\ln(\frac{\Delta}{\theta})/\beta}\leq\Delta$.
\end{lemma}

\begin{proof}[Sketch]
    Suppose $\Delta=e^{i\beta}\theta$. Then the lemma is true because $\ln(\frac{\Delta}{\theta})/\beta=i$.
\end{proof}

\addtocounter{lemma}{1}

\begin{lemma}
    $\PR[\widehat{\Delta}>e^{\mu+c}\Delta]\leq\frac{1}{2}e^{-c\epsilon_1/\beta}$ for any $c>0$.
\end{lemma}

\begin{proof}[Sketch]
    \begin{align}
        \PR[\widehat{\Delta}>e^{\mu+c}\Delta]=&\PR_{Y\sim\LAP(\beta/\epsilon_1)}[Y>c]\\
        =&\PR_{Y\sim\LAP(1)}[Y>c\epsilon_1/\beta]\\
        =&\frac{1}{2}e^{-c\epsilon_1/\beta}
    \end{align}
\end{proof}

\begin{lemma}
    $\PR[\widehat{\Delta}<\Delta]\leq\frac{1}{2}e^{-\mu\epsilon_1/\beta}$.
\end{lemma}

\begin{proof}[Sketch]
    The same as the previous lemma.
\end{proof}

\begin{lemma}
    For any fixed $\widehat{\Delta}\geq 0$, $GS_X\leq\widehat{\Delta}$.
\end{lemma}

\begin{proof}[Sketch]
    For all neighboring $(P_1,M_1)\preceq(P_2,M_2)$, let $i=\arg\min_i H_i(P_1,M_1)+(|P_1|-i)\widehat{\Delta}$ and $j=\arg\min_j H_j(P_2,M_2)+(|P_2|-j)\widehat{\Delta}$. Then, we have
    \begin{align}
        X(P_1,M_1)=&H_i(P_1,M_1)+(|P_1|-i)\widehat{\Delta}\\
        \leq& H_{j-1}(P_1,M_1)+(|P_1|-(j-1))\widehat{\Delta}\\
        \leq& H_j(P_2,M_2)+(|P_2|-j)\widehat{\Delta}\\
        =&X(P_2,M_2)\\
        X(P_2,M_2)=&H_j(P_2,M_2)+(|P_2|-j)\widehat{\Delta}\\
        \leq& H_i(P_2,M_2)+(|P_2|-i)\widehat{\Delta}\\
        \leq& H_i(P_1,M_1)+(|P_1|-i+1)\widehat{\Delta}\\
        =&X(P_1,M_1)+\widehat{\Delta}
    \end{align}
\end{proof}

\begin{lemma}
    If $\widehat{\Delta}\geq\Delta$, then $H_{|P|-g\ln(\frac{\Delta}{\theta})/\beta}\leq X\leq H_{|P|}$.
\end{lemma}

\begin{proof}[Sketch]
    The second inequality is obvious, so we show the first inequality. Let $i=\arg\min_i H_i+(|P|-i)\widehat{\Delta}$ and suppose $\Delta=e^{j\beta}\theta$, then
    \begin{align}
        H_{|P|-g\ln(\frac{\Delta}{\theta})/\beta}=&H_{|P|-gj}\\
        &(\text{by the property of $g$-bounding sequence})\notag\\
        \leq& H_i+(|P|-i)G_{|P|-j}\\
        \leq& H_i+(|P|-i)\Delta\\
        \leq& H_i+(|P|-i)\widehat{\Delta}\\
        =&X
    \end{align}
\end{proof}

\begin{theorem}
    For parameters $\epsilon_1>0$, $\epsilon_2>0$, $\beta>0$, $\theta>0$ and $\mu>0$, recursive mechanism, as described above, satisfies $(\epsilon_1+\epsilon_2)$-differential privacy, and is $(e^{2\mu}\Delta^* c/\epsilon_2+
    g\lceil \ln(\frac{\Delta^*}{\theta})/\beta\rceil G_{|P|},
    e^{-\mu\epsilon_1/\beta}+
    e^{-c})$-accurate for any $c>0$, where $\Delta^*=\max\{\theta,e^\beta G_{|P|}\}$.
    If $\epsilon_1=\Theta(\epsilon)$, $\epsilon_2=\Theta(\epsilon)$, $\beta=\epsilon_1/k$, and $\theta$ and $\mu$ are constants, then the mechanism is $(O(k\ln(G_{|P|})G_{|P|}/\epsilon),2e^{-k\mu})$-accurate as $\epsilon\rightarrow 0$,$k\rightarrow\infty$ and $G_{|P|}\rightarrow\infty$.
\end{theorem}

\begin{proof}[Sketch]
    The privacy guarantee is obvious since both the computation of $\widehat{\Delta}$ and $\widehat{X}$ satisfy differential privacy. The utility guarantee is also true because

    1) with probability at least $1-e^{-\mu\epsilon_1/\beta}$, we have $\Delta<\widehat{\Delta}<e^{2\mu}\Delta$;

    2) with probability at least $1-e^{-c}$, we have $|\widehat{X}-X|\leq \widehat{\Delta}c/\epsilon_2$;

    3) if $\widehat{\Delta}\geq\Delta$, we have $|X-H_{|P|}|\leq (g\ln(\frac{\Delta}{\theta})/\beta)G_{|P|}$
\end{proof}

\begin{theorem}
    The sequence $H$ is a recursive sequence, and the sequence $G$ is a bounding sequence of $H$.
\end{theorem}

\begin{proof}[Sketch]
    For any neighboring $(P_1,M_1)\preceq(P_2,M_2)$, $P_1\cup\{p\}=P_2$, and for any $0\leq i\leq|P_1|$, $y\in\{1,2\}$, let $P_y^i=\arg\min_{P'\subseteq P_y,|P'|=i}q(M_y(P'))$. Then, because $H_0=0$ and
    \begin{align}
        H_i(P_2,M_2)=q(M_2(P_2^i))
        &\leq q(M_2(P_1^i))\\
        &=q(M_1(P_1^i))\\
        &=H_i(P_1,M_1)\\
        H_i(P_1,M_1)=q(M_1(P_1^i))
        &\leq q(M_1(P_2^{i+1}-\{p\}))\\
        &\leq q(M_2(P_2^{i+1}))\\
        &=H_{i+1}(P_2,M_2)
    \end{align}
    $H$ is a recursive sequence. The same reasoning also applies to $G$ being a recursive sequence.

    Now we show that $G$ is a bounding sequence of $H$. For any $0\leq i\leq j\leq |P|$, let $A=\arg\min_{P'\subseteq P,|P'|=i}q(M(P'))$, and let $B=\arg\min_{P'\subseteq P,|P'|=i}\widetilde{GS}_q(P',M)$. Then
    \begin{align}
        H_j\leq& q(M(B))\\
        \leq& q(M(A\cap B))+|B-A|\widetilde{GS}_q(B,M)\\
        \leq& q(M(A))+(|P|-i)\widetilde{GS}_q(B,M)\\
        =& H_i+(|P|-i)G_j
    \end{align}
\end{proof}

For the following proofs, we define $f\cup g,f\cap g:P\rightarrow[0,1]$ by $(f\cup g)(p)=\max(f(p),g(p))$ and $(f\cap g)(p)=\min(f(p),g(p))$. We also define $f_p:P\rightarrow[0,1]$ as an indicator function that has $f_p(p)=1$ and $f_p(p')=0$ for all $p'\neq p$, and let $f_{P'}=\sum_{p\in P'}f_p$.

\begin{theorem}
    The sequence $H$ is a recursive sequence, and $H_{|P|}(P,R)=q(\SUPP(R))$.
\end{theorem}

\begin{proof}[Sketch]
    $H_{|P|}(P,R)=q(\SUPP(R))$ is obvious due to correctness of $\phi$. We will show that $H$ is a recursive sequence.

    For any neighboring $(P_1,R_1)\preceq(P_2,R_2)$, $P_1\cup\{p\}=P_2$, and for any $0\leq i\leq |P_1|$, $y\in\{1,2\}$, let $f_y^i=\arg\min_{f\in[0,1]^{P_y},|f|=i}\sum_t q(t)\phi_{R_y(t)}(f)$. Then, we have $H_0=0$ and
    \begin{align}
        H_i(P_2,R_2)=&\sum_t q(t)\phi_{R_2(t)}(f_2^i)\\
        \leq&\sum_t q(t)\phi_{R_2(t)}(f_1^i)\\
        &\text{(naturalness of $\phi$)}\notag\\
        =&\sum_t q(t)\phi_{R_2(t)_{|p\rightarrow\FALSE}}(f_1^i)\\
        &\text{($R_1(t)$ and $R_2(t)_{|p\rightarrow\FALSE}$ are equivalent)}\notag\\
        =&\sum_t q(t)\phi_{R_1(t)}(f_1^i)\\
        =&H_i(P_1,R_1)
    \end{align}
    Due to $|f_2^{i+1}\cap(1-f_p)|\geq i$ and monotonicity of $\phi$, we have
    \begin{align}
        H_i(P_1,R_1)=&\sum_t q(t)\phi_{R_1(t)}(f_1^i)\\
        \leq&\sum_t q(t)\phi_{R_1(t)}(f_2^{i+1}\cap(1-f_p))\\
        =&\sum_t q(t)\phi_{R_2(t)}(f_2^{i+1}\cap(1-f_p))\\
        \leq&\sum_t q(t)\phi_{R_2(t)}(f_2^{i+1})\\
        =&H_{i+1}(P_2,R_2)
    \end{align}
\end{proof}

\addtocounter{lemma}{1}

\begin{theorem}
    The sequence $G$ is a $2$-bounding sequence of $H$.
\end{theorem}

\begin{proof}[Sketch]
    The proof for $G$ being a recursive sequence is the same as the proof for $H$. Now we show that for any $0\leq i\leq j\leq |P|$, we have $H_j\leq H_i+(|P|-i)G_k$, where $k=|P|-\lfloor(|P|-j)/2\rfloor$.

    Let $h=\arg\min_{h\in[0,1]^P,|h|=i}\sum_t q(t)\phi_{R(t)}(h)$, $g=\arg\min_{g\in[0,1]^P,|g|=k}2\max_p\sum_t q(t)\phi_{R(t)}(g)S_{R(t),p}$, and $f(p)=\max(0,1-2(1-g(p)))$ for all $p\in P$. We first observe that, due to truncated linearity of $\phi$, if $\phi_{R(t)}(f)>0$, then $\phi_{R(t)}(g)>0.5$. Thus,
    \begin{align}
        H_j\leq&\sum_t q(t)\phi_{R(t)}(f) \quad\quad \text{(note that $|f|\geq j$)} \\
        \leq&\sum_t q(t)\phi_{R(t)}(h\cap f)+\notag\\
        &|f-h\cap f|\max_{p}\sum_{t:\phi_{R(t)}(f)>0} q(t)S_{R(t),p}\\
        &\text{($\phi_{R(t)}(f)>0\Rightarrow \phi_{R(t)}(g)>0.5$)}\\
        \leq&\sum_t q(t)\phi_{R(t)}(h\cap f)+\notag\\
        &|f-h\cap f|\max_{p}\sum_{t} 2q(t)\phi_{R(t)}(g)S_{R(t),p}\\
        \leq& \sum_t q(t)\phi_{R(t)}(h)+(|P|-i)G_k\\
        =&H_i+(|P|-i)G_k
    \end{align}
\end{proof}

\begin{theorem}
    The mapping $\phi$, defined above, have the desired properties of correctness, naturalness, monotonicity, convexity, and truncated linearity.
\end{theorem}

\begin{proof}[Sketch]
    These properties can be easily proved by induction, so we omit the details.
\end{proof}

\begin{lemma}
    (Convexity of $H$) $H_{i+1}-H_i\leq H_{i+2}-H_{i+1}$ for all $0\leq i\leq |P|-2$.
\end{lemma}

\begin{proof}[Sketch]
    First note that the function $h(f)=\sum_t q(t)\phi_{R(t)}(f)$ is convex, due to the convexity of $\phi$. Then, let $f^i=\arg\min_{f\in[0,1]^P,|f|=i}h(f)$. We have
    \begin{align}
        H_{i+1}=h(f^{i+1})
        \leq& h((f^i+f^{i+2})/2)\\
        & \text{(convexity of $h$)} \notag\\
        \leq & (h(f^i)+h(f^{i+2}))/2\\
        = & (H_i+H_{i+2})/2
    \end{align}
\end{proof}

%

\end{document}